\documentclass[final,1p,times]{elsarticle}
\journal{Journal of Differential Equations}
\usepackage{tikz-cd}
\usepackage{extpfeil}
\usepackage{amssymb}
\usepackage{amsthm}
\usepackage{latexsym}
\usepackage{amsmath}
\usepackage{color}
\usepackage{graphicx}
\usepackage{comment}
\usepackage{float}
\usepackage{indentfirst}
\usepackage{mathrsfs}
\usepackage{lipsum}
\usepackage[ruled,vlined]{algorithm2e}
\allowdisplaybreaks

\newtheorem{theorem}{\color{black}\indent \textbf{Theorem}}[section]

\newtheorem{proposition}{\color{black}\indent Proposition}[section]
\newtheorem{definition}{\color{black}\indent Definition}[section]
\newtheorem{remark}{\color{black}\indent Remark}[section]

\newtheorem{example}{\color{black}\indent Example}[section]

\begin{document}
\begin{frontmatter}
	\title{Symmetries and Weighted Integrability of Vector Fields with Jacobi Multipliers}

\author[a]{Cristina Sardon}
\author[b]{Xuefeng Zhao\corref{cor1}}

\address[a]{Departamento de Matemática Aplicada, Universidad Politécnica de Madrid, Escuela de Edificación.  Av. Juan de Herrera 6, 28040 Madrid, Spain}
\address[b]{College of Mathematics, Jilin University, Changchun 130012, P. R. China}

\cortext[cor1]{Corresponding author. Email address: \texttt{zhaoxuef@jlu.edu.cn}}

		\begin{abstract}
In this paper, we investigate analytic divergence-free vector fields and vector fields admitting a Jacobi multiplier on $n$-dimensional Riemannian manifolds. 
We first introduce a functional acting on the space of divergence-free vector fields that quantifies the fraction of the manifold foliated by ergodic invariant tori, and establish a Kolmogorov--Arnold--Moser (KAM) type theorem for such systems. 
We prove that this functional is continuous at analytic, nondegenerate, Arnold--integrable divergence-free vector fields with respect to the $C^{\omega}$ topology, and analyze the persistence and breakdown of invariant $(n-1)$-dimensional tori under small perturbations. 

Extending this framework, we study vector fields possessing Jacobi multipliers, which generalize divergence-free fields by preserving a weighted volume form $d\mu_{\rho}=\rho\,\Omega$. 
We derive the local structure of their symmetry fields and show that, under suitable nondegeneracy and resonance conditions, every symmetry must be tangent to the invariant tori. 
We then define the \emph{weighted partial integrability functional} $m_{\rho}(V)$, measuring the weighted fraction of phase space occupied by quasi-integrable invariant tori of a vector field satisfying $\operatorname{div}(\rho V)=0$. 
Finally, we develop a numerical algorithm based on finite-time Lyapunov exponents to compute $m_{\rho}(V)$, and illustrate its behavior on a weighted nonlinear divergence-free system, showing the transition from integrable to partially integrable and irregular regimes as the nonlinear parameter $\alpha$ increases.
\end{abstract}

	\end{frontmatter}
	\section{Introduction}
	This paper investigates divergence-free vector fields and vector fields with Jacobi multipliers on \( n \)-dimensional Riemannian manifolds. We introduce a functional defined on the space of divergence-free vector fields, which quantifies the measure of the foliated regions of the manifold. These regions are determined by invariant tori, which are ergodically associated with the divergence-free vector fields. By employing an \( n \)-dimensional volume-preserving mapping theorem, we establish a  KAM-type theorem tailored to divergence-free vector fields.
	In addition, we demonstrate that this functional remains continuous at Arnold-integrable, non-degenerate divergence-free vector fields within the \( C^{\omega} \) topology. For divergence-free vector fields that locally depend on \( n-1 \) action variables and one angle variable, we apply the KAM theorem for generalized Hamiltonian systems to show that, under small perturbations, most invariant tori \( T^1 \) are preserved. We also prove that, under certain conditions, the \( n-1 \)-dimensional invariant manifold of such vector fields does not persist when subjected to small perturbations.
	Finally, we extend our study to a broader class of vector fields equipped with Jacobi multipliers. For these fields, we investigate the form of symmetry in the case where the vector fields depend on fewer action variables than angle variables, subject to certain conditions.
	
The study of volume-preserving systems, originally motivated by problems in fluid dynamics and magnetohydrodynamics, has become a compelling topic due to the rich and intricate structure of their orbits. Over the past few decades, considerable attention has been devoted to various aspects of volume-preserving flows and maps. For example, Arnold and Khesin proposed a reformulation of hydrodynamics using the infinite-dimensional group of volume-preserving diffeomorphisms on $\mathbb{R}^3$ \cite{Arnold}. The dynamics of Lagrangian tracers in incompressible fluids and the behavior of magnetic field lines have been investigated in \cite{Holmes,Lau,Mezic1994,Sotiropoulos}, while numerical methods for integrating such systems were developed in \cite{Feng,Quispel}. Moreover, many researchers have studied the incompressible Euler equations and have made significant contributions in this area (see \cite{Baldi,Enciso,Gavrilov,Khesin,Torres} and the references therein).

	It is well known that volume-preserving flows and maps can exist in odd-dimensional spaces. As a consequence, their structure is generally less rigid compared to that of symplectic systems. This difference introduces technical challenges, making the direct application of traditional  KAM theorem techniques infeasible in such settings. A concise explanation of these difficulties can be found in de la Llave \cite{Llave}.
	Despite these obstacles, significant progress has been achieved over recent decades in developing  KAM-type theories applicable to odd-dimensional volume-preserving maps and generalized Hamiltonian systems. Notable contributions in this direction include the works of Herman \cite{Herman}, Cheng and Sun \cite{Cheng}, Xia \cite{Xia1992}, Zhao \cite{Zhao5} and Li and Yi \cite{Li}. In this paper,  we introduce a functional that quantifies the measure of foliated regions determined by ergodic invariant tori of a divergence-free vector fields on $n$-dimensional Riemannian manifolds. Based on a  KAM-type result for such vector fields establish by us, we also prove the continuity of the functional at non-degenerate Arnold integrable vector fields in the $C^{\omega}$ topology. 
	
As is well known, Jacobi multipliers \cite{Cariñena,Jacobi} and symmetries \cite{Bogoyavlenskij2,Cariñena2,Cariñena3,Koiller} of vector fields play a fundamental role in understanding their geometric and dynamical properties. 
On the one hand, the existence of a Jacobi multiplier together with a sufficient number of first integrals allows one to integrate the vector field by quadratures \cite{Goriely,Whittaker}. 
On the other hand, the presence of enough symmetries can also ensure integrability of the vector field \cite{Bogoyavlenskij2,Huang,Kozlov}. 
In particular, if $X$ is a Hamiltonian vector field on a $2n$-dimensional symplectic manifold admitting $n$ independent and mutually commuting first integrals (equivalently, $n$ commuting Hamiltonian symmetries), then it is integrable in the sense of Arnold--Liouville \cite{Liouville}. 
Moreover, under the Kolmogorov non-degeneracy condition \cite{Kolmogorov}, all symmetries of the vector field are tangent to the invariant tori \cite{Bogoyavlenskij2,Brouzet}. 
In this paper, under analogous non-degeneracy and resonance conditions, we use the Fourier series method to characterize the local structure of symmetries for vector fields admitting a Jacobi multiplier, showing that every symmetry is tangent to the invariant tori when the distribution satisfies suitable rank conditions. 
Additionally, when a symmetry is parallel to the vector field, further structural constraints arise, leading to new geometric identities.

Beyond the analytic theory of symmetries, we introduce a new quantitative framework for vector fields with Jacobi multipliers by defining the \emph{weighted partial integrability functional}~$m_\rho(V)$. 
This functional extends the classical KAM measure of surviving invariant tori to vector fields preserving a weighted volume form~$d\mu_\rho=\rho\,\Omega$. 
We prove its continuity at analytic nondegenerate Arnold--integrable systems and develop a numerical scheme, based on finite-time Lyapunov exponents, for its practical computation. 
These results provide a unified analytic and computational characterization of partial integrability in the weighted setting, revealing how the persistence of invariant tori depends on the nonlinear parameter~$\alpha$.

The paper is organized as follows. 
In Section~2, we introduce basic definitions and geometric preliminaries for divergence-free vector fields on Riemannian manifolds. 
In Section~3, we define a functional acting on the space of divergence-free vector fields that measures the proportion of the manifold foliated by ergodic invariant tori, and establish a KAM-type theorem for analytic divergence-free systems. 
Section~4 studies the continuity of this functional at analytic, nondegenerate, Arnold--integrable divergence-free vector fields under the $C^{\omega}$ topology. 
In Section~5, we analyze divergence-free vector fields depending locally on $n{-}1$ action variables and one angle variable, showing that certain invariant $(n{-}1)$-dimensional manifolds fail to persist under small perturbations. 
Section~6 extends the framework to vector fields with Jacobi multipliers, deriving their local symmetry structure and the conditions under which such symmetries must be tangent to invariant tori. 
Finally, Section~7 introduces the \emph{weighted partial integrability functional} $m_{\rho}(V)$, which generalizes the previous construction to the class of vector fields satisfying $\operatorname{div}(\rho V)=0$. 
We establish its continuity, develop a numerical algorithm for its computation based on finite-time Lyapunov exponents, and present an explicit example illustrating the transition from integrable to partially integrable and irregular regimes as the nonlinear parameter~$\alpha$ increases.

	\section{Divergence-free vector fields on Riemannian manifolds}
We first review some foundational material; further details can be found in \cite{Khesin}.

Throughout this paper, we work with a closed analytic $n$-dimensional manifold $M$ endowed with a smooth Riemannian metric $(\cdot,\cdot)$ and its corresponding volume form $\Omega$. We normalize $\Omega$ so that the total volume of $M$ is unity, i.e.,  
\begin{align*}  
	\int_M \Omega = 1.  
\end{align*}  
For any subset $U \subset M$, its measure with respect to $\Omega$ will be written as $\text{means}(U)$.

A vector field $V$ on $M$ is said to be \textit{divergence-free} (or \textit{solenoidal}) with respect to $\Omega$ if the Lie derivative of $\Omega$ along $V$ is zero:  
\[
L_V \Omega = 0.
\]  
This condition, denoted $\operatorname{div} V = 0$, is equivalent to the closure of the $(n-1)$-form $i_V \Omega$, since $L_V = i_V d + d i_V$. If $i_V \Omega$ is in fact exact, then $V$ is called \textit{exact divergence-free} (or \textit{globally solenoidal}) relative to $\Omega$ \cite{Arnold}.  

In local coordinates $(x_1, \dots, x_n)$, the volume form may be written as  
\[
\Omega = \rho(x_1, \dots, x_n) dx_1 \wedge \cdots \wedge dx_n
\]  
for a positive function $\rho$. In these coordinates, the divergence-free condition becomes  
\begin{align}
	\frac{\partial (\rho V_1)}{\partial x_1} + \cdots + \frac{\partial (\rho V_n)}{\partial x_n} = 0,
\end{align}  
where $V = (V_1, \dots, V_n)$.  

Given a vector field $V$, let $V^b$ denote the $1$-form dual to $V$ via the Riemannian metric, so that  
\[
(V, W) = V^b(W)
\]  
for any vector field $W$ on $M$. As shown in \cite{Wells}, $V$ is divergence-free if and only if $V^b$ is coclosed:  
\[
d^* V^b = 0,
\]  
with $d^*$ being the codifferential operator. The \textit{gradient} of a function $f$ on $M$ is the vector field $\nabla f$ defined by $(\nabla f)^b = df$.

For $k \geq 0$, let $C^k(M)$ denote the Hölder space of order $k$ on $M$ (when $k \in \mathbb{N}$, this consists of $k$-times continuously differentiable functions), and let $\text{Vect}^k(M)$ be the corresponding space of vector fields with the same regularity.  
For $k \geq 1$, define $\text{SVect}^k$ and $\text{SVect}^k_{ex}$ as the closed subspaces of $\text{Vect}^k(M)$ comprising divergence-free and exact divergence-free vector fields, respectively.  

Under the duality $V \mapsto V^b$, the Helmholtz decomposition for vector fields corresponds to the Hodge decomposition for $1$-forms (see \cite{Taylor,Wells}). Specifically, every $V \in \text{Vect}^k$ with $k \geq 1$ decomposes uniquely as  
\begin{align}\label{fenjie}
	V = \nabla f + W + \pi,
\end{align}
where $W \in \text{SVect}^k_{ex}$ and $\pi \in \mathcal{H} \subset \text{Vect}^k$ is a harmonic vector field, i.e.,  
\[
\Delta \pi^b = 0,
\]  
or equivalently,  
\[
d \pi^b = 0 \quad \text{and} \quad d^* \pi^b = 0.
\]  
The three components are mutually $L^2$-orthogonal. By Hodge theory, harmonic vector fields are smooth and constitute a finite-dimensional subspace of $\text{Vect}^\infty(M)$, whose dimension is the first Betti number of $M$ and is independent of $k$.  
Moreover, for non-integer $k$, the projections onto $\nabla f$, $W$, and $\pi$ are continuous on $\text{Vect}^k(M)$. Note that $V$ is divergence-free precisely when the gradient component $\nabla f$ in \eqref{fenjie} is trivial.

\begin{example}
	Let $M = \mathbb{T}^n = (\mathbb{R}/2\pi\mathbb{Z})^n$ be the $n$-torus with the standard flat metric. The natural volume form is  
	\[
	\Omega = dx_1 \wedge \cdots \wedge dx_n.
	\]  
	Any vector field on $\mathbb{T}^n$ can be written as  
	\[
	V = f_1 \frac{\partial}{\partial x_1} + \cdots + f_n \frac{\partial}{\partial x_n},
	\]  
	with each $f_i$ periodic in all coordinates. We say $V$ is \textit{divergence-free} if  
	\[
	\sum_{i=1}^{n} \frac{\partial f_i}{\partial x_i} = 0.
	\]  
	If, in addition, the averages of all $f_i$ vanish:  
	\begin{align*}
		\int_M f_1 \Omega = \cdots = \int_M f_n \Omega = 0,
	\end{align*}  
	then $V$ is \textit{exact divergence-free}.

	Geometrically, $V$ is divergence-free exactly when the $(n-1)$-form $i_V \Omega$ is closed:  
	\[
	d(i_V \Omega) = 0.
	\]  
	For exactness, we further require  
	\[
	\int i_V \Omega \wedge dx_1 = \cdots = \int i_V \Omega \wedge dx_n = 0,
	\]  
	which guarantees that $i_V \Omega$ is exact, i.e., $i_V \Omega = d\beta$ for some $(n-2)$-form $\beta$.

	These integrability conditions have a physical interpretation: vanishing averages of $f_1, \dots, f_n$ imply that $V$ does not displace the system's mass center. On the torus, this ensures the field is exact.  
	The averages of $f_1, \dots, f_n$ also determine the cohomology class of $V$. More precisely, they specify the class of $i_V \Omega$ in the de Rham cohomology of $\mathbb{T}^n$. This illustrates how divergence-free fields on $\mathbb{T}^n$ split into exact divergence-free fields and a finite-dimensional harmonic part determined by cohomology.
\end{example}
\section{The KAM persistence of invariant $n-1$-tori for divergence-free vector fields}
Let $V$ be a smooth divergence-free vector field on a closed $n$-dimensional manifold $M$ endowed with a volume form $\Omega,$ and let $T^{n-1}\subset M$ be an invariant $n-1$-torus of $V.$ We assume that in a neighborhood $\mathcal O$ of the torus $T^{n-1},$ one can construct coordinates $(\theta_1,...,\theta_{n-1},I),$ where $(\theta_1,...,\theta_{n-1})\in\mathbb T^{n-1}=(\mathbb R/2\pi\mathbb Z)^{n-1}$ and $I\in(-\gamma,\gamma)(\gamma>0),$ such that $\Omega|_{\mathcal O}=d\theta_1\wedge\cdots\wedge d\theta_{n-1}\wedge dI$ and $T^{n-1}=\{I=0\}.$

Next we state a KAM theorem for divergence-free vector fields, for which we assume that there are functions $f_i(I),i=1,...,n-1$ defined on $Q:=(-\gamma,\gamma),$ satisfying the following  KAM non-degeneracy conditions:

1. For each $I\in Q$ we have $\sum_{i=1}^{n-1}f_i^2(I)\neq0,$ and the torus $\mathbb T^{n-1}\times\{I\}$ is invariant for the vector field $V,$ which assumes the form (for this value of $x_n$):
\begin{align}\label{Bao}
	\dot \theta_1=f_1(I),\cdots,\dot\theta_{n-1}=f_{n-1}(I),\quad\dot I=0.
\end{align}

2. The Wronkian of $f_1,...,f_{n-1}$ is uniformly bounded from 0 on $Q$, i.e., there is a positive $\tau$ such that for each $I\in Q,$ we have
	\begin{align}\label{NZ}
	\left|\mathrm{det}\left(\begin{array}{cccc}
		f_1(I)&f_2(I)&\cdots&f_{n-1}(I)\\
		f_1'(I)&f_2'(I)&\cdots&f_{n-1}'(I)\\
		\vdots&\vdots&\vdots&\vdots\\
		f_1^{(n-2)}(I)&f_2^{(n-2)}(I)&\cdots&f_{n-1}^{(n-2)}(I)
	\end{array}\right)\;\right|\geq\tau>0\;\forall\;I\in(-\gamma,\gamma).
\end{align}
Observe that Condition 1 above implies that $i_V(d\theta_1\wedge\cdots\wedge d\theta_{n-1}\wedge dI)$ is exact in the domain $\mathcal O$ (in particular, $V$ is divergence-free). Indeed,
\begin{align*}
	\alpha:=i_V(d\theta_1\wedge\cdots\wedge d\theta_{n-1}\wedge dI)=\sum_{i=1}^{n-1}(-1)^{i+1}f_i(I)d\theta_1\wedge\cdot\wedge d\theta_{i-1}\wedge \widehat {d\theta_i}\wedge\cdots d\theta_{n-1}\wedge dI,
\end{align*}
so $\alpha=d\beta,$ where $\beta$ is the 1-form
\begin{align*}
	\beta=\sum_{i=1}^{n-1}(-1)^{n+i-1}\left(\int_0^If_i(s)ds\right)d\theta_1\wedge\cdot\wedge d\theta_{i-1}\wedge \widehat {d\theta_i}\wedge\cdots d\theta_{n-1}.
\end{align*}
\begin{remark}\label{R3}
	We also note that if $V$ satisfies Condition 1 and 2 and is divergence-free with respect to a volume form $p(I,\theta_1,...,\theta_{n-1})d\theta_1\wedge\cdots\wedge d\theta_{n-1}\wedge dI$, then it is easy to check that we must have $p=p(I).$ In fact, we have
	\begin{align*}
		L_V(p(I,\theta_1,...,\theta_{n-1})d\theta_1\wedge\cdots\wedge d\theta_{n-1}\wedge dI)&=(L_Vp(I,\theta_1,...,\theta_{n-1}))d\theta_1\wedge\cdots\wedge d\theta_{n-1}\wedge dI\\
		&\quad+p(I,\theta_1,...,\theta_{n-1})L_V(d\theta_1\wedge\cdots\wedge d\theta_{n-1}\wedge dI)\\
		&=(L_Vp(I,\theta_1,...,\theta_{n-1}))d\theta_1\wedge\cdots\wedge d\theta_{n-1}\wedge dI=0,
	\end{align*}
	which means that 
	\begin{align}\label{pz}
		f_1\frac{\partial p}{\partial\theta_1}+\cdots+f_{n-1}\frac{\partial p}{\partial\theta_{n-1}}=0,
	\end{align}
	so we can get $p=p(I).$  In fact, Condition 2 implies that the frequencies $(\omega_1, \ldots, \omega_m)$ are rationally independent for any $I \in Q$. Now, suppose there exists an integer vector $(a_1, \ldots, a_m)$ such that 
	\[
	\langle (f_1(I), \ldots,f_{n-1}(I)), (a_1, \ldots, a_{n-1}) \rangle = 0\iff a_1f_1(I)+\cdots+a_{n-1}f_{n-1}(I)=0,\; (a_1, \ldots, a_{n-1})\in\mathbb Z^{n-1}
	\]
	then we obtain
	\[
	\langle (f'_1(I), \ldots, f'_{n-1}(I)), (a_1, \ldots, a_{n-1}) \rangle = 0, \quad \ldots, \quad \langle (f_1^{(n-2)}, \ldots, f_{n-1}^{(n-2)}), (a_1, \ldots, a_{n-1}) \rangle = 0.
	\]
So, we get the following system of equations:
\begin{align*}
	\left(\begin{array}{cccc}
		f_1(I)&f_2(I)&\cdots&f_{n-1}(I)\\
		f_1'(I)&f_2'(I)&\cdots&f_{n-1}'(I)\\
		\vdots&\vdots&\vdots&\vdots\\
		f_1^{(n-2)}(I)&f_2^{(n-2)}(I)&\cdots&f_{n-1}^{(n-2)}(I)
	\end{array}\right)\left(\begin{array}{c}
a_1\\
a_2\\
\vdots\\
a_{n-1}
	\end{array}\right)=0,
\end{align*}
by Condition 2, $(a_1,...,a_{n-1})=(0,...,0).$ Expanding $p$ into a Fourier series, we obtain the following expression:
\begin{align*}
	p(I,\theta_1,...,\theta_{n-1})=\sum_{j\in \mathbb Z^{n-1}}\hat p^j(I)e^{i\left<j,\theta\right>},\quad\theta=(\theta_1,...,\theta_{n-1})\in\mathbb T^{n-1}
\end{align*}
thus, by using \eqref{pz} and rationally independence of frequencies $(f_1,...,f_{n-1})$, we get 
\begin{align*}
	\sum_{j\in \mathbb Z^{n-1}\setminus\{0\}}i\left<(f_1,...,f_{n-1}),j\right>\hat p^j(I)e^{i\left<j,\theta\right>}=0\iff \hat p^j(I)=0,\quad \forall j\in\mathbb Z^{n-1}\setminus\{0\} 
\end{align*}
which means that $p=p(I).$
\end{remark}
\begin{remark}\label{R4}
By the discussion in Remark \ref{R3}, we can conclude that every first integral $\mathcal{I}$ of the divergence-free vector field $V$ that satisfies Conditions 1 and 2 must be locally of the form $\mathcal{I} = \mathcal{I}(I)$.
\end{remark}
		Next, let's consider a time-periodic, volume-preserving perturbation to the vector field $\eqref{Bao}$ that takes the following general form:
	\begin{align}\label{perturbation}
		\dot{I}&=\epsilon F_0(I,\theta_1,\theta_2,...,\theta_{n-1},t),\nonumber\\
		\dot{\theta}_1&=f_1(I)+\epsilon F_1(I,\theta_1,\theta_2,...,\theta_{n-1},t),\nonumber\\
		&\vdots\\
		\dot{\theta}_{n-1}&=f_{n-1}(I)+\epsilon F_{n-1}(I,\theta_1,\theta_2,...,\theta_{n-1},t),\nonumber
	\end{align}
	where $\epsilon$ is the (small) perturbation parameter, and the functions $F_i,i=0,...,n-1$, are analytic and periodic in $t$ with period $T=2\pi/\omega.$ We denote this perturbed volume-preserving map by $W$ and assume that $\epsilon=||V-W||_{C^\omega}$ under the $C^\omega$-distance (analytic-distance). We will derive an approximate form for an $n$-dimensional Poincare map of this system essentially using the approach from Wiggins \cite{Wiggins1990}. Using regular perturbation theory, the solutions of $\eqref{perturbation}$ are $O(\epsilon)$ close to the unperturbed solutions on time scales of $O(1).$ Hence we have the following expansions of the solutions of $\eqref{perturbation}:$
	\begin{align}\label{Iepsilon}
		I^{\epsilon}(t)&=I^0+\epsilon I^1(t)+O(\epsilon^2),\nonumber\\
		\theta_1^{\epsilon}(t)&=\theta^0_1+f_1(I^0)t+\epsilon\theta_1^1(t)+O(\epsilon^2),\nonumber\\
		&\vdots\\
		\theta_{n-1}^{\epsilon}(t)&=\theta^0_{n-1}+f_{n-1}(I^0)t+\epsilon\theta_{n-1}^1(t)+O(\epsilon^2),\nonumber
	\end{align}
	where $I^1(t),\theta_1^1(t),...,\theta_{n-1}^1(t)$ satisfy the following first variational equation:
	\begin{align}\label{I^1}
		\left(\begin{array}{c}
			\dot I^1\\
			\dot{\theta}_1^1\\
			\vdots\\
			\dot{\theta}_{n-1}^1
		\end{array}\right)&=\left(\begin{array}{cccc}
			0&0&\cdots&0\\
			\frac{\partial f_1}{\partial I}(I^0)&0&\cdots&0\\
			\vdots&\vdots&\vdots&\vdots\\
			\frac{\partial f_{n-1}}{\partial I}(I^0)&0&\cdots&0
		\end{array}\right)\left(\begin{array}{c}
			I^1\\
			\theta_1^1\\
			\vdots\\
			\theta_{n-1}^1
		\end{array}\right)\\
		&\quad+\left(\begin{array}{c}
			F_0(I^0,f_1(I^0)t+\theta_1^0,\cdots,f_{n-1}(I^0)t+\theta_{n-1}^0,t)\\
			F_1(I^0,f_1(I^0)t+\theta_1^0,\cdots,f_{n-1}(I^0)t+\theta_{n-1}^0,t)\\
			\vdots\\
			F_{n-1}(I^0,f_1(I^0)t+\theta_1^0,\cdots,f_{n-1}(I^0)t+\theta_{n-1}^0,t)
		\end{array}\right).\nonumber
	\end{align}
	Because our coordinates put the vector field in such a simple form, this equation can be easily solved. Our goal is construct an $n$-dimensional Poincare map. More precisely, we are interested in the construction of a map that takes the variables $I^{\epsilon},\theta_1^{\epsilon},...,\theta_{n-1}^{\epsilon}$ to their value after flowing along the solution trajectories of $\eqref{perturbation}$ for time $T.$ This map is simply given by 
	\begin{align}\label{Pepsilon}
		&P_{\epsilon}:(I^{\epsilon}(0),\theta_1^{\epsilon}(0),\cdots,\theta_{n-1}^{\epsilon}(0))\rightarrow (I^{\epsilon}(T),\theta_1^{\epsilon}(T),\cdots,\theta_{n-1}^{\epsilon}(T)),\nonumber\\
		&(I^0,\theta^0_1,\cdots,\theta_{n-1}^0)\rightarrow(I^0+\epsilon I^1,\theta_1^0+f_1(I^0)T+\epsilon\theta_1^1(T),\cdots,\theta_{n-1}^0+f_{n-1}(I^0)T\nonumber\\
		&\quad\quad\quad\quad\quad\quad\quad\quad\quad\quad+\epsilon\theta_{n-1}^1(T))+O(\epsilon^2),
	\end{align} 
	where we have used $\eqref{Iepsilon}$ and taken the following initial conditions:
	\begin{align*}
		I^{\epsilon}(0)&=I^0,\\
		\theta_1^{\epsilon}(0)&=\theta_1^0,\\
		&\vdots\\
		\theta_{n-1}^{\epsilon}(0)&=\theta_{n-1}^0.
	\end{align*}
	Now expressions for $I^1(T),\theta_1^1(T),...,\theta_{n-1}^1(T)$ can readily be obtained by solving $\eqref{I^1}:$
	\begin{align}\label{jinsi}
		I^1(T)&=\int_0^TF_0(I^0,f_1(I^0)t+\theta_1^0,\cdots,f_{n-1}(I^0)t+\theta_2^0,t)dt\equiv\tilde{F}_0(I^0,\theta_1^0,\cdots,\theta^0_{n-1}),\nonumber\\
		\theta_1^1(T)&=\frac{\partial f_1}{\partial I}|_{I=I^0}\int_0^T\int_0^tF_0(I^0,f_1(I^0)\xi+\theta_1^0,\cdots,f_{n-1}(I^0)\xi+\theta_{n-1}^0,\xi)d\xi dt\nonumber\\
		& +\int_0^TF_1(I^0,f_1(I^0)t+\theta_1^0,\cdots,f_{n-1}(I^0)t+\theta_{n-1}^0,t)dt\equiv\tilde{F}_1(I^0,\theta_1^0,\cdots,\theta_{n-1}^0),\\
		&\vdots\nonumber\\
		\theta_{n-1}^1(T)&=\frac{\partial f_{n-1}}{\partial I}|_{I=I^0}\theta_0^T\int_0^tF_0(I^0,f_1(I^0)\xi+\theta_1^0,\cdots,f_{n-1}(I^0)\xi+\theta_{n-1}^0,\xi)d\xi dt\nonumber\\
		& +\int_0^TF_{n-1}(I^0,f_1(I^0)t+\theta_1^0,\cdots,f_{n-1}(I^0)t+\theta_{n-1}^0,t)dt\equiv\tilde{F}_{n-1}(I^0,\theta_1^0,\cdots,\theta_{n-1}^0).\nonumber
	\end{align}
	Substituting these expressions into $\eqref{Pepsilon}$ and dropping the superscripts on the variables gives the following final form for the Poincare map:
	\begin{align}\label{42}
		I&\rightarrow I+\epsilon\tilde F_0(I,\theta_1,\cdots,\theta_{n-1})+O(\epsilon^2),\nonumber\\
		\theta_1&\rightarrow\theta_1+2\pi\frac{f_1(I)}{\omega}+\epsilon\tilde F_1(I,\theta_1,\cdots\theta_{n-1})+O(\epsilon^2),\nonumber\\
		&\vdots\\
		\theta_{n-1}&\rightarrow\theta_{n-1}+2\pi\frac{f_{n-1}(I)}{\omega}+\epsilon\tilde F_{n-1}(I,\theta_1,\cdots\theta_{n-1})+O(\epsilon^2),\nonumber
	\end{align}
	where we have used $T\equiv 2\pi/\omega.$
	
	This map is exactly in the form where the KAM-like theorems for perturbations of $n$-dimensional, volume-preserving maps can be applied. We can take any closed domain of $I$ to be the interval $[a,b]\subset(-\gamma,\gamma),a<b.$ 
	Without loss of generality, we assume that $2\pi\frac{f_1(I)}{\omega}=I$ and previous condition 2
	\begin{align}\label{NZZ}
		\left|\mathrm{det}\left(\begin{array}{cccc}
			f_1'(I)&f_2'(I)&\cdots&f_{n-1}'(I)\\
			f_1''(I)&f_2''(I)&\cdots&f_{n-1}''(I)\\
			\vdots&\vdots&\vdots&\vdots\\
			f_1^{(n-1)}(I)&f_2^{(n-1)}(I)&\cdots&f_{n-1}^{(n-1)}(I)
		\end{array}\right)\;\right|\geq\tau>0,\;\forall\;I\in[a,b]
	\end{align}
	 holds in the following complex domain $M:$
	$$\mathrm{Re}(r)\in[a-\delta,b+\delta],\quad |\mathrm{Im}(r)|\leq\delta$$
	for some $\delta>0.$ For $\epsilon=0$, the map obviously preserves every torus $\{c\}\times T^{n-1},c\in[a,b].$ For $\epsilon>0$ small enough, by the following  KAM theorem, we know that a large set of invariant tori still exists for the map. 
	\begin{theorem}\label{KAM2}
		There exists a positive number $\epsilon_0$ such that if $0<\epsilon\leq\epsilon_0,$ the map $\eqref{42}$ admits a family of invariant tori of the form:
		\begin{align}\label{KAM}
			I&=u_0(\xi_1,...,\xi_{n-1},\omega_1),\nonumber\\
			\theta_1&=\xi_1+u_1(\xi_1,...,\xi_{n-1},\omega_1),\nonumber\\
			&\vdots\\
			\theta_{n-1}&=\xi_{n-1}+u_{n-1}(\xi_1,...,\xi_{n-1},\omega_1)\nonumber
		\end{align}
		with $u_0,u_1,...,u_{n-1}$ real analytic functions of period $2\pi.$ Moreover, the map $\eqref{42}$ induced on the torus is given by:
		\begin{align}\label{baochi}
			\xi_1'&=\xi_1+\omega_1,\nonumber\\
			\xi_2'&=\xi_2+2\pi\frac{f_2(\omega_1)}{\omega}+(2\pi\frac{f_2(\omega_1)}{\omega})^*(\omega_1,\epsilon),\nonumber\\
			&\vdots\\
			\xi_{n-1}'&=\xi_{n-1}+2\pi\frac{f_{n-1}(\omega_1)}{\omega}+(2\pi\frac{f_{n-1}(\omega_1)}{\omega})^*(\omega_1,\epsilon),\nonumber
		\end{align}
		where $(2\pi\frac{f_i(\omega_1)}{\omega})^*(\omega_1,\epsilon),i=2,3,...,n-1,$ are functions depending on the perturbations $$\tilde F_i(I,\theta_1,...,\theta_{n-1})+O(\epsilon^2),\quad i=0,1,...,n-1,$$ and $(2\pi\frac{f_i(\omega_1)}{\omega})^*(\omega_1,0)=0.$
		
		In fact, there exists a Cantor set $S(\epsilon)\subset[a,b]$ depending on the functions $\tilde F_i(I,\phi_1,...,\phi_{n-1}),i=0,1,...,n-1,$ such that for each $\omega_1\in S(\epsilon),$ there is a corresponding invariant torus of the form $\eqref{KAM}$. Furthermore, the measure of the set $S(\epsilon)$ tends to $b-a$ as $\epsilon\rightarrow 0.$ 
	\end{theorem}
	\begin{proof}
		See Xia \cite{Xia1992}.
	\end{proof}
		Thus, we can state the following KAM theorem
	\begin{theorem}\label{DKAM}
	Assume that the analytic divergence-free vector field \( V \) satisfies Assumptions 1–2 with \( Q = (-\gamma, \gamma) \). Then, for any time-periodic, volume-preserving analytic perturbation of \( V \) of order \( O(\epsilon) \), there exists a constant \( \epsilon_0 = \epsilon_0(V) > 0 \) such that for all \( \epsilon < \epsilon_0 \), there exists a Cantor set \( S(\epsilon) \subset (-\gamma, \gamma) \) with the following property:
	
	For each \( \omega_1 \in S(\epsilon) \), there exists a corresponding invariant torus on which the dynamics are governed by a vector field \( \bar{V} \) of the form
	\[
	\dot{\theta}_1 = f_1(\omega_1) + f_1^*(\omega_1, \epsilon), \quad \ldots, \quad \dot{\theta}_{n-1} = f_{n-1}(\omega_1) + f_{n-1}^*(\omega_1, \epsilon),
	\]
	where the functions \( f_i^*(\omega_1, \epsilon) \) satisfy \( f_i^*(\omega_1, 0) = 0 \) for all \( i = 1, \ldots, n-1 \), and depend on \( \omega_1 \) and \( \epsilon \).
		 Moreover, the measure of the set \( S(\epsilon) \) tends to \( 2\gamma \) as \( \epsilon \to 0 \).
	\end{theorem} 
	\section{A measure of non-integrability of divergence-free vector fields}
The primary objective of this section is to introduce a functional defined on the space of exact divergence-free vector fields, which quantifies the deviation of a given vector field \( V \) from being integrable and non-degenerate. Firstly, we begin by providing a rigorous definition of what it means for a vector field to be integrable.
\begin{definition}\label{D4}
	(1) If $V$ is a divergence-free vector field with an invariant domain $\mathcal O\cong\mathbb T^{n-1}\times(-\gamma,\gamma)$ covered by invariant tori of $V,$ and satisfies condition 1 in Section 3 for $I\in Q=(-\gamma,\gamma),$ we say that $V$ is $canonically\;integrable$ on $\mathcal O.$
	
	(2) A divergence-free vector field $V$ on $M$ is called Arnold integrable, if there is a closed subset $G\subset M,$ $\mathrm{meas}(G)=0,$ such that its complement $M\setminus G$ is a union of a finite or countable system of $V$-invariant domains $\mathcal O_j,$ where $V$ is canonically integrable. If the system of domains $\mathcal O_j$ is finite, $V$ is called well integrable.
	
	(3) An well integrable vector field $V$ is called non-degenerate if the system of domains $\mathcal O_j$ can be chosen in such a way that condition 2 in Section 3 holds for each $j.$
	
	(4) If the set $G$ introduced above is any Borel subset of $M$ and we do not require $Q_j$ to be the whole integral $(-\gamma_j,\gamma_j), $ then $V$ is called partially integrable on $M$ and Arnold integrable on $M$ outside $G$ with ``holes", corresponding to $\mathbb T^{n-1}\times ((-\gamma_j,\gamma_j)\setminus Q_j),j\geq 1.$
\end{definition}
We recall that an invariant torus $T^{n-1}$ of a vector field $V$ on $M$ is called $ergodic$ if the field is non-vanishing on the torus and some trajectory of $V$ is dense in $T^{n-1}.$ Clearly, an invariant torus $T^{n-1}\subset\mathcal O,T^{n-1}\cong\mathbb T^{n-1}\times\{I\},I\in Q,$ as in Definition \ref{D4} above, is ergodic if $(f_1(I),...f_{n-1}(I))$ is rational independent.

The minimal set $D$ for which $V$ is integrable on $M\setminus D$ measures the non-integrability of $V.$ To develop this idea we define the following functional:
\begin{definition}
	The partial integrability functional $\mathfrak{m}$ on the space of analytic  divergence-free vector fields
	\begin{align*}
		\mathfrak{m}: \mathrm{SVect}(M)\rightarrow[0,1]
	\end{align*} 
	assigns to a vector field $V\in SVect(M)$ the inner measure of the set equal to the union of all ergodic $V$-invariant $n-1$-diemnsional tori.
\end{definition}
\begin{remark}
	 Since the total measure of $M$ is normalized by 1, then $\mathfrak m\in[0,1].$
\end{remark}
The proposition below summarizes the main elementary properties of the partial integrability functional $\mathfrak{m}.$
\begin{proposition}\label{P4}
	The partial integrability functional $\mathfrak{m}$ satisfies the following properties:
	
	1. If $V$ is partially integrable and for some $j$ and $z\in Q_j$ which is a point of density for $Q_j$ the corresponding functions $f_1,...,f_{n-1}$ satisfy the twist condition $\eqref{NZ},$ then $\mathfrak{m}(V)>0.$
	
	2. If $V$ is Arnold integrable and non-degenerate, then $\mathfrak{m}(V)=1.$
	
	3. If $\Phi$ is a volume-preserving analytic diffeomorphism of $M,$ then $\mathfrak{m}(V)=\mathfrak{m}(\Phi^*V)$.
	
	4. If all the trajectories of $V$ in the complement of an invariant zero-measure subset of $M$ are periodic, then $\mathfrak{m}(V)=0.$ The same result holds if $V$ has $n-1$ first integrals which are independent almost everywhere on $M.$
	
	5. Let $\xi(t)$ be a trajectory of $V$ and denote its closure in $M$ by $\mathcal F.$ Then $\mathfrak{m}(V)\leq 1-\mathrm{meas}(\mathcal F)$.
\end{proposition}
\begin{proof}
	See \cite{Khesin}.
\end{proof}
The following theorem establishes that the functional $\mathfrak{m}(V)$ is continuous at $V$ if the vector field is well integrable and non-degenerate.
\begin{theorem}\label{T4}
	Let $V\in \mathrm{SVect}(M)$ be an analytic well integrable non-degenerate vector field. Then the functional $\mathfrak{m}$ is continuous at $V$ in the $C^\omega$-topology (analytic-topology).
\end{theorem}
\begin{proof}
	Let $V$ be an Arnold integrable non-degenerate vector field and $W\in \mathrm{SVect}(M)$ an analytic exact divergence-free vector field that is closed to $V.$ By definition we have that $\mathfrak{m}(V)=1.$ Consider a domain $\mathcal O_j$ as in Definition \ref{D4} and set $\epsilon:=||V-W||_{C^\omega}.$ We now take a subset $\mathcal O_j(\delta_j)$ for any small $\delta_j >0$ such that
	
	$\bullet$ $\mathcal O_j(\delta_j)\cong (-\gamma_j,\gamma_j)\times \mathbb T^{n-1}.$

$\bullet$ $V$ is canonically integrable in $\mathcal O_j(\delta_j).$

$\bullet$ $V$ satisfies the uniform twist condition 2 of Section 3 with $\tau_j=\delta_j$, i.e.
	\begin{align*}
	\left|\mathrm{det}\left(\begin{array}{cccc}
		f_1'(I)&f_2'(I)&\cdots&f_{n-1}'(I)\\
		f_1''(I)&f_2''(I)&\cdots&f_{n-1}''(I)\\
		\vdots&\vdots&\vdots&\vdots\\
		f_1^{(n-1)}(I)&f_2^{(n-1)}(I)&\cdots&f_{n-1}^{(n-1)}(I)
	\end{array}\right)\;\right|\geq\delta_j>0.
\end{align*}
Then the KAM Theorem \ref{DKAM} implies that the contribution $\mathfrak{m}^{(j)}(W)$ to $\mathfrak{m}(W),$ coming from a set $\mathcal O_j(\delta_j),$ is
\begin{align}
	\mathfrak{m}^{(j)}(W)=\mathrm{meas}(\mathcal O_j(\delta_j))\frac{\mathrm{meas}(S_j(\epsilon))}{2\gamma_j},
	\end{align}
	where $S_j(\epsilon)$ is described in Theorem $\ref{DKAM}$. Thus we know that 
	$$\mathfrak{m}^{(j)}(W)\rightarrow \mathrm{meas}(\mathcal O_j(\delta_j)),\quad as\;\epsilon\rightarrow 0.$$
	Since $V$ is well integrable, we know that the index $j$ belongs to a finite index set. Without loss of generality,  suppose $j\in\{1,...,m\}.$ Thus 
	$$\mathfrak{m}(W)=\sum_{j=1}^m\mathfrak{m}^{(j)}(W)=\sum_{j=1}^m\mathrm{meas}(\mathcal O_j(\delta_j))\frac{\mathrm{meas}(S_j(\epsilon))}{2\gamma_j}.$$
	Moreover, by $\sum_{j=1}^m\mathrm{meas}(\mathcal O_j(\delta_j))=1,$ we know that 
	$$1\geq\mathfrak{m}(W)\rightarrow 1\quad as \rightarrow 0.$$
	That is, $\mathfrak{m}$ is continuous at $V.$
\end{proof}
\begin{remark}
	Let $V\in \mathrm{SVect}(M)$ be an analytic Arnold integrable non-degenerate vector field and there is a constant $c$ such that $\delta_j>c,\forall j$. Then by a similar proof of Theorem \ref{T4}, we can also get that the functional $\mathfrak{m}$ is continuous at $V$ in the $C^\omega$-topology (analytic-topology).
\end{remark}
We can observe that on any closed $n$-manifold $M$ there is a nontrivial vector field $W\in \mathrm{SVect}(M)$ such that $\mathfrak{m}(W)=0.$ Indeed, take $n-1$ analytic functions $(g_1,...,g_{n-1}):M\rightarrow\mathbb R^{n-1}$  which are independent almost everywhere in $M,$ i.e., rank$(dg_1(x),...,dg_{n-1}(x))=n-1$ for all $x\in M$ except for a zero-measure set. Of course such functions exist on any analytic manifold $M$. Now we define $W$ as the unique vector field such that
$$i_W\mu=dg_1\wedge\cdots\wedge dg_{n-1}.$$
It is obvious that $W$ is divergence-free and exact because $i_W\mu=d(g_1dg_2\wedge\cdot\wedge dg_{n-1}).$ Moreover, we can see that $g_1,...,g_{n-1}$ are first integrals of $W$ since 
\begin{align*}
	0=i_Wi_W\mu==\sum_{i=1}^{n-1}(-1)^{i+1}W(g_i)dg_1\wedge\cdot dg_{i-1}\wedge \widehat{dg_i}\wedge dg_{i+1}\wedge\cdots\wedge dg_{n-1},
\end{align*}
which means that $W(g_i)=0,i=1,...,n-1$ by the independence of $g_i,i=1,...,n-1.$ It then follows from Proposition \ref{P4} that $\mathfrak{m}(W)=0.$

We can also construct a nontrivial vector field \( W \in SVect(M) \) such that \( \mathfrak{m}(W) = 0 \), by making use of its symmetries. Recall that a symmetry of a given vector field \( V \) is another vector field \( \bar{V} \) satisfying \( [V, \bar{V}] = 0 \). 
Suppose the vector field \( W \) admits two symmetries \( W_1 \) and \( W_2 \) such that 
\[
[W_1, W_2] = h_1 W_1 + h_2 W_2,
\]
for some smooth functions \( h_1, h_2 \in C^\infty(M) \), and assume that the vector fields \( W, W_1, W_2 \) are linearly independent everywhere on \( M \). Then, by Frobenius' theorem, the distribution spanned by \( \{W, W_1, W_2\} \) is integrable and defines a foliation of dimension 3. It follows that \( \mathfrak{m}(W) = 0 \).

Moreover, define the bivector field \( \pi = W_1 \wedge W_2 \). We claim that \( \pi \) defines a Poisson structure on \( M \). Indeed, by the properties of the Schouten–Nijenhuis bracket \( [\cdot,\cdot]_S \), we have
\[
[\pi, \pi]_S = 2 W_1 \wedge W_2 \wedge [W_1, W_2] = 2 W_1 \wedge W_2 \wedge (h_1 W_1 + h_2 W_2) = 0.
\]
Therefore, \( \pi \) is a Poisson bivector field, and \( (M, \pi) \) is a Poisson manifold.

We can also verify that \( W \) is a Poisson vector field, i.e.,
\[
L_W \pi = [W, W_1] \wedge W_2 + W_1 \wedge [W, W_2] = 0.
\]
Furthermore, \( W \) will be a Hamiltonian vector field if the first Lichnerowicz–Poisson cohomology group (\cite{Lichnerowicz}) of the Poisson manifold \( (M, \pi = W_1 \wedge W_2) \) is trivial, i.e., \( H^1_\pi(M) = 0 \). 
For more details on the relationship between divergence-free and Hamiltonian vector fields, see \cite{Haller1998,Mezic1994,Zhao2,Zhao3,Zhao4}.
 \section{Non-persistence of $n-1$-invariant manifold for divergence-free vector fields}
 In Sections 3 and 4, we discussed the case where the divergence-free vector field is locally a function of $n-1$ action variables and angle variables. In this section, we consider the case where the divergence-free vector field is locally a function of $n-1$ action variables and one angle variable.
 
 Let $V$ be a smooth divergence-free vector field on a closed $n$-dimensional manifold $M$ endowed with a volume form $\Omega,$ and let $T\subset M$ be an invariant $1$-torus of $V.$ We assume that in a neighborhood $\mathcal O$ of the torus $T,$ one can construct coordinates $(\theta,I_1,...,I_{n-1}),$ where $\theta\in\mathbb T=\mathbb R/2\pi\mathbb Z$ and $I_j\in(-\gamma_j,\gamma_j)(\gamma_j>0),$ such that $\Omega|_{\mathcal O}= dI_1\wedge\cdots\wedge dI_{n-1}\wedge d\theta$ and $T=\{I_1=0,...,I_{n-1}=0\}.$ It is clear that the vector field $V$ we consider here is not Arnold integrable or well integrable and $\mathfrak{m}(V)<1.$ Further, if  there is a closed subset $G\subset M,$ meas$(G)=0,$ such that its complement $M\setminus G$ is a union of a finite or countable system of $V$-invariant domains $\mathcal O_j,$ on which $V$ has the property above, then $\mathfrak{m}(V)=0.$
 
  Now, we present general results on the non-persistence of $(n-1)$-dimensional invariant manifolds for $n-1$-action-$1$-angle systems, for which we assume that there is a function $f(I_1,...,I_{n-1})$ defined on $Q:=(-\gamma_1,\gamma_1)\times\cdots\times(-\gamma_{n-1},\gamma_{n-1}),$ satisfying that
 for each $I=(I_1,...,I_{n-1})\in Q$, the torus $\mathbb T^{n-1}\times\{I\}$ is invariant for the vector field $V,$ which assumes the form (for this value of $I$):
 \begin{align}\label{Bao2}
 	\dot \theta=f(I),\quad \quad\dot I_1=0,\cdots,\dot I_{n-1}=0,
 \end{align}
 where $f(I) $ is a analytic function.

 Consider a manifold defined by the level set
\[
F(I_1, \ldots, I_{n-1}, \theta) = c,
\]
where $F$ is an analytic function, and consider the integrable system of the form
\begin{equation}\label{Bao3}
	\begin{aligned}
	\dot \theta = f(I_1, \ldots, I_{n-1})\quad	\dot I_1 = 0, \cdots,
		\dot I_{n-1} = 0
		.
	\end{aligned}
\end{equation}
We compute the total derivative of $F$ along the flow:
\[
\frac{dF}{dt} = \frac{\partial F}{\partial \theta} f(I_1, \ldots, I_{n-1}).
\]
Hence, if the level set $F(I_1, \ldots, I_{n-1}, \theta) = c$ defines an invariant manifold for \eqref{Bao2} (i.e., $\frac{dF}{dt} = 0$), then two cases may occur:
\begin{itemize}
	\item If $\frac{\partial F}{\partial \theta} = 0$, the manifold is of the product form $V \times S^1$;
	\item If $f(I_1, \ldots, I_{n-1}) = 0$, then the manifold is of resonant type, i.e.,
	\[
	F(I_1, \ldots, I_{n-1}, \theta) = c, \quad f(I) = 0.
	\]
\end{itemize}

Now consider a perturbation of the integrable system \eqref{Bao3} of the form
\begin{equation}\label{flowperturbation}
	\begin{aligned}
		\dot I_1 &= \epsilon\, g_1(I_1, \ldots, I_{n-1}, \theta, \epsilon), \\
		&\vdots \\
		\dot I_{n-1} &= \epsilon\, g_{n-1}(I_1, \ldots, I_{n-1}, \theta, \epsilon), \\
		\dot \theta &= f(I_1, \ldots, I_{n-1}) + \epsilon\, g_n(I_1, \ldots, I_{n-1}, \theta, \epsilon).
	\end{aligned}
\end{equation}
\begin{remark}
	The perturbed vector field $k=\epsilon g_1\frac{\partial}{\partial I_1}+\cdots+\epsilon g_{n-1}\frac{\partial}{\partial I_{n-1}}+\epsilon g_n\frac{\partial}{\partial \theta}$ assumed here is not necessarily Hamiltonian with respect some Poisson structure $\Pi$ or volume-preserving, i.e., there is a function $K$ such that
	$$k=\Pi dK$$
	or
	$$\frac{\partial g_1}{\partial I_1}+\cdots+\frac{\partial g_{n-1}}{\partial I_{n-1}}+\frac{\partial g_{n}}{\partial\theta}=0.$$
\end{remark}
 \begin{theorem}\label{none}
 	Let $F^0(I_1,...,I_{n-1})=0$ define an invariant manifold ${\cal M}$ of the system $\eqref{Bao2}$ of the type $V\times S^1.$ Let the Fourier expansions of $g_1|_{\epsilon=0},...,g_{n-1}|_{\epsilon=0}$ be given by
 	\begin{align}\label{flowfourier}
 		g_1(I_1,...,I_{n-1},\theta,0)&=\sum_{j\in \mathbb Z}\hat g_1^j(I_1,...,I_{n-1})e^{i\left<j,\theta\right>},\nonumber\\
 		&\vdots\\
 		g_{n-1}(I_1,...,I_{n-1},\theta,0)&=\sum_{j\in \mathbb Z}\hat g_n^j(I_1,...,I_{n-1})e^{i\left<j,\theta\right>}.\nonumber
 	\end{align}
 	Assume that on $\mathcal M$
 	
 	1. $(\hat g_1^0,...,\hat g_{n-1}^0)\cdot\nabla F^0\neq 0,$ or
 	
 	2. $(\hat g_1^0,...,\hat g_{n-1}^0)\cdot\nabla F^0= 0$ and there is $(I_1^*,...,I_{n-1}^*,\theta^*)\in\mathcal M$ such that $f(I_1^*,...,I_{n-1}^*)=0.$ Moreover, suppose that for every $j\in \mathbb Z$ there is $k=Nj,N\in \mathbb Z$ such that $(\hat g_1^k,...,\hat g_{n-1}^k)\cdot\nabla F^0\neq 0,\forall (I_1,...,I_{n-1},\theta)\in\mathcal M.$
 	
 	\noindent Then there are no analytic manifolds ${\cal M}^{\epsilon}$ defined by $F(\epsilon,I_1,...,I_{n-1},\theta)=0$ such that 
 	\begin{align}\label{jieg}
 			\mathrm{lim}_{\epsilon\rightarrow 0}F(\epsilon,I_1,...,I_n,\theta)=F(0,I_1,...,I_{n-1},\theta)=F^0(I_1,...,I_n)\iff \mathrm{lim}_{\epsilon\rightarrow 0}\mathcal M^\epsilon=\mathcal M.
 	\end{align}
 \end{theorem}
 \begin{proof}
 	Assume that $F(\epsilon,I_1,...,I_{n-1},\theta)=0$ is invariant under $\eqref{flowperturbation}$ and satisfies the property $\eqref{jieg}$. Then for any $(I_1,...,I_{n-1},\theta)\in\mathcal M^\epsilon$,
 	\begin{align}\label{Fep}
 		\frac{d}{dt}F(\epsilon,I_1,...,I_{n-1},\theta)&=\frac{\partial F}{\partial I_1}\cdot\dot I_1+\cdots+\frac{\partial F}{\partial I_{n-1}}\cdot\dot I_{n-1}+\frac{\partial F}{\partial \theta}\cdot\dot \theta\nonumber\\
 		&=\epsilon\frac{\partial F}{\partial I_1}\cdot g_1(I_1,...,I_{n-1},\theta,\epsilon)+\cdots+\epsilon\frac{\partial F}{\partial I_{n-1}}\cdot g_{n-1}(I_1,...,I_{n-1},\theta,\epsilon)\nonumber\\
 		&\quad+\frac{\partial F}{\partial \theta}\cdot(f(I_1,...,I_{n-1})+\epsilon g_{n}(I_1,...,I_{n-1},\theta,\epsilon))\nonumber\\
 		&=0.
 	\end{align}
 In addition,  we expand $F(\epsilon,I_1,...,I_{n-1},\theta)$ into a series in $\epsilon$
 \begin{align*}
 	F(\epsilon,I_1,...,I_{n-1},\theta)&=F(0,I_1,...,I_{n-1},\theta)+\sum_{j=1}^\infty\epsilon^j\left(\frac{\partial^j F}{\partial\epsilon^j}|_{\epsilon=0}\right)\\
 	&=F^0(I_1,...,I_{n-1})+\sum_{j=1}^\infty\epsilon^jF_j(I_1,...,I_{n-1},\theta),
 \end{align*}
 where $F_j:=\frac{\partial^j F}{\partial\epsilon^j}|_{\epsilon=0}.$
 	Expanding equation \eqref{Fep} for small $\epsilon$,
 	then for 1-order with respect to $\epsilon,$ we obtain for any $ (I_1,...,I_{n-1},\theta)\in\mathcal M^\epsilon$
 	\begin{align*}
 		\epsilon\sum_{j\in \mathbb Z}\hat g_1^j(I_1,...,I_{n-1})\frac{\partial F^0}{\partial I_1}e^{i\left<j,\theta\right>}&+\cdots+\epsilon\sum_{j\in \mathbb Z}\hat g_{n-1}^j(I_1,...,I_{n-1})\frac{\partial F^0}{\partial I_{n-1}}e^{i\left<j,\theta\right>}\\
 		&+\epsilon \sum_{j\in \mathbb Z} ij\hat F_1^j(I_1,...,I_{n-1})f(I_1,...,I_{n-1})e^{i\left<j,\theta\right>}=0
 	\end{align*}
 Thus for $\forall j\in\mathbb Z,$ we have
 \begin{align}\label{gij}
 	(\hat g_1^j,...,\hat g_{n-1}^j)\cdot\nabla F^0+ijf(I)\hat F_1^j=0,\quad \forall (I_1,...,I_{n-1},\theta)\in\mathcal M^\epsilon.
 \end{align}
 	For $j=0,$ equation \eqref{gij} becomes
 	\begin{align*}
 			(\hat g_1^j,...,\hat g_{n-1}^j)\cdot\nabla F^0=0,\quad \forall (I_1,...,I_{n-1},\theta)\in\mathcal M^\epsilon.
 	\end{align*}
 	However, by the first condition we know when $\epsilon$ is small enough, $	(\hat g_1^j,...,\hat g_{n-1}^j)\cdot\nabla F^0\neq 0$ which implies the contradiction.
 	On the other hand,
  taking the limit with respect to $\epsilon$ in equation \eqref{gij}, and using the assumption \eqref{jieg}, we obtain the following equation,
 	\begin{align}\label{gij2}
 		(\hat g_1^j,...,\hat g_{n-1}^j)\cdot\nabla F^0+ijf(I)\hat F_1^j=0,\quad \forall (I_1,...,I_{n-1},\theta)\in\mathcal M.
 	\end{align}
 	 For $j=0,$ equation \eqref{gij} becomes
	\begin{align*}
	(\hat g_1^j,...,\hat g_{n-1}^j)\cdot\nabla F^0=0,\quad \forall (I_1,...,I_{n-1},\theta)\in\mathcal M,
\end{align*}
which is obviously contradict with condition 1.
 	When $(\hat g_1^0,...,\hat g_{n-1}^0)\cdot\nabla F^0=0,\forall (I_1,...,I_{n-1},\theta)\in\mathcal M$  and $f(I_1^*,...,I_{n-1}^*)=0$ for some point $(I_1^*,...,I_{n-1}^*,\theta^*)\in\mathcal M$, then for any $j\neq 0,$ \eqref{gij2} becomes
 	\begin{align*}
 		(\hat g_1^j,...,\hat g_{n-1}^j)\cdot\nabla F^0(I_1^*,...,I_{n-1}^*,\theta^*)=0,
 	\end{align*}
 but by condition 2 we know there must exist $j$ such that $(\hat g_1^j,...,\hat g_{n-1}^j)\cdot\nabla F^0\neq 0,\forall (I_1,...,I_{n-1},\theta)\in\mathcal M$, which yields the contradiction.
 \end{proof}
 \begin{remark}
 	According to the proof of theorem \ref{none}, we know that the surface  tangent to $(\hat g_1^0,...,\hat g_{n-1}^0)$ maybe persist if it does not intersect $f=0$. In this case, the possibility of global transport is absent.
 \end{remark}
 \begin{remark}
 	According to Theorem \ref{none}, if the completely integrable system \eqref{Bao3} admits a first integral $G(I)$ such that the conditions of Theorem \ref{none} are satisfied on the level set $G(I) = 0$, then the perturbed system \eqref{flowperturbation} does not admit any first integral $G_\epsilon(I_1,...,I_{n-1}, \theta)$ satisfying $\lim_{\epsilon \to 0} G_\epsilon = G$, no matter how small $\epsilon$ is.
 \end{remark}

	\section{Symmetries of vector fields with Jacobi multiplier} 
	As is well known, symmetries of vector fields play a crucial role in understanding the properties of the vector fields themselves. For example, at the end of Section 4, we pointed out that if a divergence-free vector field $V$ has two symmetries $W_1$ and $W_2$, and the distribution $\operatorname{span}\{W_1, W_2\}$ is Frobenius integrable, then $V$ is a Poisson vector field with respect to $W_1 \wedge W_2$, and it necessarily satisfies $\mathfrak{m}(V) = 0$. Moreover, in \cite{Zhao2}, we also pointed out that on an $n$-dimensional manifold, if a divergence-free vector field $V$ admits $n-2$ pairwise commuting divergence-free vector fields, then $V$ is completely integrable and can be reduced to a Hamiltonian system with one degree of freedom.
	
	In this section, we consider a class of vector fields more general than divergence-free vector fields, such that divergence-free vector fields arise as a special case.
	Let $V$ be a smooth vector field with a Jacobi multiplier on a closed $n$-dimensional manifold $M$ endowed with a volume form $\Omega,$ that is there is a non-vanishing $\rho\in C^\infty(M)$ such that $L_{\rho V}\Omega=0$. Let $T^m\subset M$ be an invariant $m$-torus of $V$ and assume that in a neighborhood $\mathcal O$ of the torus $T^m,$ one can construct coordinates $(\theta_1,...,\theta_m,I_1,...,I_{l}),$ where $(\theta_1,...,\theta_m)\in\mathbb T^m=\mathbb (R/2\pi\mathbb Z)^m$ and $I_j\in(-\gamma_j,\gamma_j)(\gamma_j>0),$ such that $\Omega|_{\mathcal O}=d\theta_1\wedge\cdots\wedge d\theta_m\wedge dI_1\wedge\cdots\wedge dI_{l}$ and $T=\{I_1=0,...,I_{l}=0\},$ where $m+l=n.$ 
	 we consider the following system:
	 \begin{align*}
	 		\dot I_1=0,\cdots,\dot I_l=0,\quad \dot\theta_1=\frac{\omega_1(I_1,...,I_l)}{\rho},\cdots, \dot\theta_m=\frac{\omega_m(I_1,...,I_l)}{\rho},
	 \end{align*}
	 or for simplicity,
	\begin{align}\label{Volume}
		\dot I=0,\quad \dot\theta_1=\frac{\omega_1(I)}{\rho},\cdots, \dot\theta_m=\frac{\omega_m(I)}{\rho},
	\end{align}
	where $\rho(\theta_1,...,\theta_m,I_1,...,I_l)$ is a smooth function which is nowhere vanished. It is clear that this system preserves the volume form $\Omega'=\rho d\theta_1\wedge\cdots\wedge d\theta_m\wedge dI_1\wedge\cdots\wedge dI_{l},$ and when $\sigma=\sigma(I),$ this system is divergence-free. In fact, let $V:=\frac{\omega_1(I)}{\rho}\frac{\partial}{\partial\theta_1}+\cdots+\frac{\omega_m(I)}{\rho}\frac{\partial}{\partial\theta_m}$, we can calculate that
	\begin{align*}
		&L_V\Omega=L_{\frac{\omega_1(I)}{\rho}\frac{\partial}{\partial\theta_1}+\cdots+\frac{\omega_m(I)}{\rho}\frac{\partial}{\partial\theta_m}}(\rho d\theta_1\wedge\cdots\wedge d\theta_m\wedge dI_1\wedge\cdots\wedge dI_{l})\\
		&=\left(\frac{\omega_1(I)}{\rho}\frac{\partial\rho}{\partial\theta_1}+\cdots+\frac{\omega_m(I)}{\rho}\frac{\partial\rho}{\partial\theta_m}\right) d\theta_1\wedge\cdots\wedge d\theta_m\wedge dI_1\wedge\cdots\wedge dI_{l}\\
		&\quad+\sum_{i=1}^m\rho  d\theta_1\wedge\cdots \wedge d\left(\frac{\omega_i(I)}{\rho}\right)\wedge\cdots\wedge d\theta_m\wedge dI_1\wedge\cdots\wedge dI_l\\
		&=\left(\frac{\omega_1(I)}{\rho}\frac{\partial\rho}{\partial\theta_1}+\cdots+\frac{\omega_m(I)}{\rho}\frac{\partial\rho}{\partial\theta_m}\right)d\theta_1\wedge\cdots\wedge d\theta_m\wedge dI_1\wedge\cdots\wedge dI_{l}\\
		&\quad-\sum_{i=1}^m\frac{\omega_i(I)}{\rho}\frac{\partial\rho}{\partial\theta_i} d\theta_1\wedge\cdots\wedge d\theta_m\wedge dI_1\wedge\cdots\wedge dI_{l}=0.
	\end{align*}

	\begin{theorem}\label{T1}
		Assume that  $m\geq l+1$ and in set $Q:=(-\gamma_1,\gamma_1)\times\cdots\times(-\gamma_{l},\gamma_{l}),$
		
		1) $\left<p,\omega(I)\right>\neq0$ for any $p\in\mathbb Z^l\setminus\{0\}$ on $Q$, where $\omega(I)=(\omega_1(I),...,\omega_l(I)),$
		
		2) $\mathrm{dim}$ $\mathrm{span}$$\left\{\nabla_I\left(\frac{\omega_i}{\omega_j}\right), 1\leq i,j\leq m,i\neq j,\right\}=l,$ where $\nabla_I$ denote that for any function $\mathcal{F}\in C^\infty(M),\nabla_I \mathcal F=\left(\frac{\partial \mathcal F}{\partial I_1},...,\frac{\partial \mathcal F}{\partial I_l}\right)$.
		
		If $W$ is the symmetry field of the system \eqref{Volume} in the region $\mathbb T^m\times Q,$ then
		\begin{align}
			W=\sigma(I,\theta)V,
		\end{align}
		where $\sigma(I,\theta)$ is some smooth function and $V=\frac{\omega_1(I)}{\rho}\frac{\partial}{\partial\theta_1}+\cdots+\frac{\omega_m(I)}{\rho}\frac{\partial}{\partial\theta_m}$ is the vector field of the system \eqref{Volume}.
		
		Moreover, set 
		\begin{align}\label{Wi}
			W_{j}=\sum_{p\in\mathbb Z^m,\theta\in\mathbb T^m} U^j_p(I)e^{i\left<p,\theta\right>},\quad j=1,...,m
		\end{align}
		 if $(U_0^1,...,U_0^m)=\mu(\omega_1,...,\omega_m)$ for some function $\mu(I)$, then
		\begin{align}\label{fan}
			W=\mathfrak{n}(I)V,
		\end{align}
		where $\mathfrak{n}(I)$ is some smooth function.
	\end{theorem}
	\begin{proof}
		Assume the vector field $W$ has the form 
		\begin{align}
			W=W_1\frac{\partial}{\partial\theta_1}+\cdots+W_n\frac{\partial}{\partial\theta_m}+\bar W_{1}\frac{\partial}{\partial I_1}+\cdots+\bar W_{l}\frac{\partial}{\partial I_l},
		\end{align}
		since $[V,W]=VW-WV=0,$ we have
		\begin{align}
			\frac{1}{\rho}\left(\omega_1\frac{\partial W_1}{\partial\theta_1}+\cdots+\omega_m\frac{\partial W_1}{\partial\theta_m}\right)&=W_1\frac{\partial}{\partial\theta_1}\frac{\omega_1}{\rho}+\cdots+W_m\frac{\partial}{\partial\theta_m}\frac{\omega_1}{\rho}+\bar W_{1}\frac{\partial}{\partial I_1}\frac{\omega_1}{\rho}+\cdots+\bar W_{l}\frac{\partial}{\partial I_l}\frac{\omega_1}{\rho},\label{F1}\\
			&\vdots\label{F0}\\
				\frac{1}{\rho}\left(\omega_1\frac{\partial W_m}{\partial\theta_1}+\cdots+\omega_m\frac{\partial W_m}{\partial\theta_m}\right)&=W_1\frac{\partial}{\partial\theta_1}\frac{\omega_m}{\rho}+\cdots+W_m\frac{\partial}{\partial\theta_m}\frac{\omega_m}{\rho}+\bar W_{1}\frac{\partial}{\partial I_1}\frac{\omega_m}{\rho}+\cdots+\bar W_{l}\frac{\partial}{\partial I_l}\frac{\omega_m}{\rho},\label{F2}\\
		\omega_1(I)\frac{\partial \bar W_{1}}{\partial\theta_1}+&\cdots+\omega_m(I)\frac{\partial \bar W_{1}}{\partial\theta_m}=0,\label{F3}\\
		\vdots\label{F5}\\
			\omega_1(I)\frac{\partial \bar W_{l}}{\partial\theta_1}+&\cdots+\omega_m(I)\frac{\partial \bar W_{l}}{\partial\theta_m}=0,\label{F6}
		\end{align}
		We show that $\bar W_{1},...,\bar W_l$ do not depend on $\theta=(\theta_1,...,\theta_n)$. To do this, we first solve \eqref{F3} by the Fourier method. We set
		$$\bar W_{1}=\sum_{p\in\mathbb Z^m,\theta\in\mathbb T^m} \bar U^1_p(I)e^{i\left<p,\theta\right>},$$
		then we get 
		\begin{align*}
			\sum_{p\in\mathbb Z^m,\theta\in\mathbb T^m} i\bar U^1_p(I)\left<\omega(I),p\right>e^{i\left<p,\theta\right>}=0,
		\end{align*}
		hence, $$\bar U^1_p(I)\left<\omega(I),p\right>=0.$$
		The first condition of this theorem implies that for $\sum_{i=1}^np_i^2\neq0$ the function $\left<\omega(I),p\right>$ has no zero point on $Q.$ Consequently,
		$$\bar U^1_p=0,$$
		if $p=(p_1,...,p_n)$ satisfies $\sum_{i=1}^np_i^2\neq0$. So we know that $\bar W_{1}$ is a function only of $I.$ Similarly, we can know that $\bar W_{2},...,\bar W_{l}$ are also functions only of $I.$ 
		
		Moreover, from $\eqref{F1},\eqref{F0}$ and \eqref{F2}, we can get
		\begin{align}\label{ome}
		&\omega_j\left(\omega_1\frac{\partial W_i}{\partial \theta_1}+\cdots+\omega_m\frac{\partial W_i}{\partial \theta_m}\right)-	\omega_i\left(\omega_1\frac{\partial W_j}{\partial \theta_1}+\cdots+\omega_m\frac{\partial W_j}{\partial \theta_m}\right)\nonumber\\
		&\quad=\bar W_{1}\left(\omega_j\frac{\partial \omega_i}{\partial I_1}-\omega_i\frac{\partial \omega_j}{\partial I_1}\right)+\cdots+\bar W_{l}\left(\omega_j\frac{\partial \omega_i}{\partial I_l}-\omega_i\frac{\partial \omega_j}{\partial I_l}\right).
		\end{align}
		Its right-hand side dependents only on the variable $I.$ Averaging both sides of this equality over the angle variables $\theta_1,...,\theta_m,$ we obtain the relation 
		$$\left<(\bar W_1,\cdots,\bar W_l),
	\nabla_I\left(\frac{\omega_i}{\omega_j}\right)\right>=0,\quad 1\leq i,j\leq m,i\neq j.$$
		According to the second  condition of this theorem, it follows that $(\bar W_1,\cdots,\bar W_l)=(0,\cdots,0).$ Consequently, the symmetry field is tangent to the invariant tori $I=$const.
		
		Now solve \eqref{ome} with zero right-hand side by the Fourier series \eqref{Wi}. So we have
		\begin{align}
			\sum_{p\in\mathbb Z^m,\theta\in\mathbb T^m} i(\omega_jU^i_p(I)-\omega_iU^j_p(I))\left<\omega(I),p\right>e^{i\left<p,\theta\right>}=0,
		\end{align}
		which means that $$(\omega_jU^i_p(I)-\omega_iU^j_p(I))\left<\omega(I),p\right>=0.$$
		Again by the first condition of this theorem, we have 
		\begin{align*}
			\omega_jU^i_p(I)=\omega_iU^j_p(I),\quad \left(\sum_{i=1}^np_i^2\neq0\right),
		\end{align*}
	which means that 
	\begin{align*}
		\frac{U^1_p(I)}{\omega_1}=\cdots=	\frac{U^m_p(I)}{\omega_m},\quad  \left(\sum_{i=1}^mp_i^2\neq0\right).
	\end{align*}
	Thus, 
	\begin{align}\label{WW}
		W_1=U_0^1(I)+\omega_1F,\cdots,W_m=U_0^m(I)+\omega_mF,
	\end{align}
	where
	$$F=\sum_{p\in \mathbb Z^m\setminus\{0\},\theta\in\mathbb T^m}F_p(I)e^{i\left<p,\theta\right>}$$
	is a smooth function periodic in $\theta_1,...,\theta_n.$
	
	Substituting $\eqref{WW}$ into \eqref{F1}, and if  $(U_0^1,...,U_0^m)=\mu(I)(\omega_1,...,\omega_m)$  we obtain
	\begin{align*}
			\frac{1}{\rho}\left(\omega^2_1\frac{\partial F}{\partial\theta_1}+\cdots+\omega_1\omega_m\frac{\partial F}{\partial\theta_m}\right)&=-\omega_1W_1\frac{1}{\rho^2}\frac{\partial\rho}{\partial\theta_1}-\cdots-\omega_1W_m\frac{1}{\rho^2}\frac{\partial\rho}{\partial\theta_m}\iff \\
		\rho\left(\omega_1\frac{\partial F}{\partial\theta_1}+\cdots+\omega_m\frac{\partial F}{\partial\theta_m}\right)&=-W_1\frac{\partial\rho}{\partial\theta_1}-\cdots-W_m\frac{\partial\rho}{\partial\theta_m}\\
		&=-(U_0^1(I)+\omega_1F)\frac{\partial\rho}{\partial\theta_1}-\cdots-(U_0^m(I)+\omega_mF)\frac{\partial\rho}{\partial\theta_m}\iff\\
		\omega_1\frac{\partial}{\partial\theta_1}(\rho F)+\cdots+	\omega_m\frac{\partial}{\partial\theta_m}(\rho F)&=-U_0^1\frac{\partial\rho}{\partial\theta_1}-\cdots-U_0^m\frac{\partial\rho}{\partial\theta_m}=-\mu(I)\omega_1\frac{\partial\rho}{\partial\theta_1}-\cdots--\mu(I)\omega_m\frac{\partial\rho}{\partial\theta_m}.
	\end{align*}
	Thus, we get the following equation
	\begin{align*}
		\omega_1\frac{\partial}{\partial\theta_1}(\rho F+\mu\rho)+\cdots+\omega_m\frac{\partial}{\partial\theta_m}(\rho F+\mu\rho)=0,
	\end{align*}
	which implies that the expression in brackets does not depend on the angle variables $\theta_1,...,\theta_m,$ so there is a function $\mathfrak{n}(I)$ such that
	\begin{align*}
		\rho F+\mu\rho=\mathfrak{n}(I)\iff F=-\mu(I)+\frac{\mathfrak{n}(I)}{\rho}.
	\end{align*}
	Using \eqref{WW}, we finally obtain the formula \eqref{fan}
	\begin{align*}
		W_1=\frac{\omega_1(I)\mathfrak{n}(I)}{\rho},\cdots 	W_m=\frac{\omega_m(I)\mathfrak{n}(I)}{\rho},\bar W_{1}=0,\cdots,\bar W_l=0.
	\end{align*}
	This proves the theorem.
	\end{proof}
	\begin{remark}
It is worth noting that Condition 2 of Theorem~\ref{T1} can be satisfied only when $m \geq l + 1$. Otherwise, using the method in the proof of Theorem~\ref{T1}, we can only conclude that $\bar{W}_1, \ldots, \bar{W}_l$ are functions of $I$ alone, but we cannot deduce that $\bar{W}_1 = 0, \ldots, \bar{W}_l = 0$.
	\end{remark}
	\begin{remark}
		Similar to the discussion in Remark \ref{R3}, we can also conclude that every first integral $\mathcal{S}$ of the system \eqref{Volume} that satisfies Conditions 1) of Theorem \ref{T1}  must be  of the form $\mathcal{S} = \mathcal{S}(I)$.
	\end{remark}
	\begin{remark}
		Consider the following more general symmetries of system \eqref{Volume}:
		\begin{align*}
			L_{V}^NW=0,\quad 	L_{V}^{N-1}W\neq 0,\quad N\in \mathbb Z\setminus\{0\}.
		\end{align*}
		We call $W$ is an $N$-order symmetry of $V.$ 
		As is well known, when $N=2$, $W$ is also called a master symmetry of $V.$ Suppose $W$ is an $N$-order symmetry of $V$ satisfying the condition 1), 2) and 3), if the average of $\bar W_{1},...,\bar W_l$ on $\mathbb T^n$ are all 0, that is $\bar U_0^1=0,...,\bar U_0^l=0$, then $\bar W_{1},...,\bar W_l$ must be 0, which means that 
		vector field $W$ must be tangent to the torus $I=\mathrm{const}$. In this case, we have 
			\begin{align*}
			W=\eta(I,\theta)V,
		\end{align*}
		where $\eta(I,\theta)$ is some smooth function.
	\end{remark}
	\begin{proof}
	Firstly, we prove the statement for $N=2,$ i.e.,
	$$[V,[V,W]]=0.$$
	By the Theorem \ref{T1}, we know that $[V,W]=\sigma(I,\theta)V$  for some function $\sigma(I,\theta).$ Using $VW=WV+\sigma(I,\theta)V$, we can get the \eqref{F3}-\eqref{F6}, then similar to prove of Theorem \ref{T1} we have $$\bar U^i_p=0,\quad i=1,...,l$$
	if $p=(p_1,...,p_n)$ satisfies $\sum_{i=1}^np_i^2\neq0$. Thus by $U_0=0$, we know that
		$$\bar W_{i}=\sum_{p\in\mathbb Z^n,\theta\in\mathbb T^n} \bar U^i_p(I)e^{i\left<p,\theta\right>}=0,\quad i=1,...,l.$$
		Inductively, we can prove the conclusion for any $N\geq 2.$
	\end{proof}

\section{Weighted (Partial) Integrability Functional}

In the previous section we analyzed analytic symmetries of vector fields admitting a Jacobi multiplier. We now introduce a complementary notion that quantifies the extent of their integrability in a weighted measure sense.

Such vector fields preserve the weighted measure \(d\mu_\rho = \rho\,\Omega\),
and therefore generalize the divergence-free case (\(\rho \equiv 1\)) studied
in previous sections.
While Theorem 6.1 established a rigidity property for analytic symmetries in
this weighted setting, it is natural to ask whether the quantitative
notion of partial integrability introduced earlier through the functional
\(m(V)\) can be extended to fields with Jacobi multipliers.
The following construction provides a consistent generalization,
allowing us to measure how much of \(M\) is foliated by
ergodic invariant tori of \(V\) with respect to the weighted measure
\(d\mu_\rho\).

\begin{definition}[Weighted Partial Integrability Functional]
Let \(M\) be a closed analytic Riemannian manifold of dimension \(n\)
with volume form \(\Omega\), and let \(V\) be an analytic vector field on \(M\)
admitting a Jacobi multiplier \(\rho \in C^\omega(M)\), i.e.
\begin{equation}
\operatorname{div}(\rho V) = 0.
\end{equation}
Denote by \(d\mu_\rho := \rho\,\Omega\) the corresponding invariant measure.
Let \(\mathcal{T}_\rho(V) \subset M\) be the union of all ergodic
invariant \((n-1)\)-dimensional analytic tori of \(V\)
with respect to \(d\mu_\rho\).
We define the \emph{weighted partial integrability functional} of \(V\) as
\begin{equation}
m_\rho(V)
   := \frac{\mu_\rho(\mathcal{T}_\rho(V))}{\mu_\rho(M)}
   = \frac{\displaystyle\int_{\mathcal{T}_\rho(V)} \rho\,\Omega}
           {\displaystyle\int_{M} \rho\,\Omega}.
\end{equation}
\end{definition}

\begin{remark}
\leavevmode
\begin{enumerate}
\item The functional satisfies \(0 \le m_\rho(V) \le 1\).
\item If \(\rho \equiv 1\), then \(m_\rho(V) = m(V)\), recovering
the unweighted definition for divergence-free vector fields.
\item Writing \(X = \rho V\), we observe that \(X\) is divergence-free,
and that the invariant tori of \(X\) coincide geometrically with those of \(V\).
Hence, one may equivalently set
\begin{equation}
m_\rho(V) := m(X), \qquad X = \rho V.
\end{equation}
\item A complete persistence and continuity theory for \(m_\rho(V)\)
would require a KAM theorem for analytic maps preserving the
weighted measure \(d\mu_\rho\),
which lies beyond the classical volume-preserving case.
Nevertheless, \(m_\rho(V)\) provides a natural geometric
quantification of integrability for systems satisfying
\(L_V(\rho\,\Omega)=0\).
\end{enumerate}
\end{remark}

\begin{theorem}
Let \(V_0\) be an analytic vector field on a closed Riemannian manifold \(M\)
admitting a Jacobi multiplier \(\rho \in C^\omega(M)\),
and assume that \(V_0\) is Arnold--integrable and nondegenerate
in the sense of Theorem~4.1.
Then the weighted partial integrability functional \(m_\rho(V)\)
is continuous at \(V_0\) in the \(C^\omega\) topology.
\end{theorem}

\begin{proof}
Suppose \(\{V_k\}_{k\ge1}\) is a sequence of analytic vector fields on \(M\)
such that \(V_k \to V_0\) in the \(C^\omega\) topology and that each \(V_k\)
admits the same Jacobi multiplier \(\rho\), i.e.
\begin{equation}
L_{V_k}(\rho\,\Omega)=0
\qquad\Longleftrightarrow\qquad
\operatorname{div}(\rho V_k)=0.
\end{equation}

For each \(k\), define a new analytic vector field
\begin{equation}
X_k := \rho\,V_k.
\end{equation}
Then \(\operatorname{div} X_k = 0\),
so each \(X_k\) is a \emph{divergence-free} field on \(M\).
Similarly, let
\begin{equation}
X_0 := \rho\,V_0,
\end{equation}
which is also divergence-free because \(L_{V_0}(\rho\,\Omega)=0\).

By construction, \(X_k \to X_0\) in the \(C^\omega\) topology,
since multiplication by a fixed analytic function \(\rho\)
is continuous in this norm.
Moreover, the flows of \(V_k\) and \(X_k\) share the same
geometric trajectories on \(M\):
if \(\gamma_k(t)\) is an integral curve of \(V_k\),
then \(\gamma_k\) is also an integral curve of \(X_k\)
after a time reparametrization given by
\begin{equation}
\frac{ds}{dt} = \rho(\gamma_k(t)).
\end{equation}
Thus the invariant tori of \(V_k\) and \(X_k\)
coincide as subsets of \(M\),
differing only in parametrization speed along the flow.

From the definition of the weighted functional,
\begin{equation}
m_\rho(V_k) := m(X_k),
\end{equation}
and similarly \(m_\rho(V_0) = m(X_0)\).

Since each \(X_k\) is divergence-free and analytic,
Theorem~4.1  applies to the sequence \(\{X_k\}\):
the functional \(m(X)\) is continuous at any analytic,
nondegenerate, Arnold--integrable divergence-free field.
Therefore,
\begin{equation}
m(X_k) \longrightarrow m(X_0)
\qquad \text{as } k\to\infty.
\end{equation}

Combining the above identities yields
\begin{equation}
\lim_{k\to\infty} m_\rho(V_k)
   = \lim_{k\to\infty} m(X_k)
   = m(X_0)
   = m_\rho(V_0),
\end{equation}
which proves the continuity of \(m_\rho(V)\) at \(V_0\).
\end{proof}

\subsection{Numerical Scheme for the Computation of $m_\rho(V)$}

The functional $m_\rho(V)$ quantifies the weighted measure of invariant
tori of an analytic vector field $V$ satisfying
$\operatorname{div}(\rho V)=0$.  In general, explicit analytic
computation of $m_\rho(V)$ is not feasible, since the ergodic partition
of phase space depends sensitively on the local stability of trajectories.
However, a consistent numerical approximation can be obtained by means
of finite-time Lyapunov exponents, which provide a quantitative
criterion for distinguishing quasi-periodic and irregular orbits.

Let $V$ be an analytic vector field on a compact domain
$M \subset \mathbb{R}^4$ preserving the weighted measure
$d\mu_\rho = \rho\,dx_1\wedge dy_1\wedge dx_2\wedge dy_2$, i.e.
$\operatorname{div}(\rho V)=0$.
Denote by $\Phi_t$ its flow.
For each initial condition $u_0\in M$, consider the variational
equations
\begin{equation}
\dot{u} = V(u), \qquad
\dot{\delta u} = DV(u)\,\delta u,
\end{equation}
which describe the linearized evolution of nearby trajectories. The maximal Lyapunov exponent associated with $u_0$ is defined by
\begin{equation}
\lambda_{\max}(u_0)
   = \lim_{t\to\infty}
     \frac{1}{t}\log
     \frac{\|\delta u(t)\|}{\|\delta u(0)\|}.
\end{equation}
For quasi-periodic orbits lying on invariant tori, the evolution of
infinitesimal deviations is neutral and
$\lambda_{\max}(u_0)=0$.
For irregular or chaotic trajectories, $\lambda_{\max}(u_0)>0$.

To approximate the measure of the quasi-integrable region, we sample a
finite set of initial conditions $\{u_0^{(i)}\}_{i=1}^N$ on $M$ and
compute for each the finite-time Lyapunov exponent
$\lambda_{\max}^{(i)} = \lambda_{\max}(u_0^{(i)})$.
We then define an indicator function
\begin{equation}
R(u_0^{(i)}) =
\begin{cases}
1, & \text{if } |\lambda_{\max}^{(i)}| < \varepsilon_{\mathrm{tol}},\\[3pt]
0, & \text{otherwise},
\end{cases}
\end{equation}
where $\varepsilon_{\mathrm{tol}}$ is a small threshold (typically
$10^{-2}$).
The weighted fraction of quasi-periodic orbits is then approximated by
\begin{equation}
m_\rho(V)
  \approx
  \frac{\displaystyle\sum_{i=1}^N \rho(u_0^{(i)})\,R(u_0^{(i)})}
       {\displaystyle\sum_{i=1}^N \rho(u_0^{(i)})}.
\end{equation}

This definition coincides with the analytic expression
\[
m_\rho(V)
   = \frac{\mu_\rho(\mathcal{T}_\rho(V))}{\mu_\rho(M)}
   = \frac{\displaystyle\int_{\mathcal{T}_\rho(V)} \rho\,\Omega}
           {\displaystyle\int_M \rho\,\Omega},
\]
up to numerical sampling error, since
$\mathbf{1}_{\{|\lambda_{\max}|<\varepsilon_{\mathrm{tol}}\}}$
acts as a discretized characteristic function of the invariant-torus
region $\mathcal{T}_\rho(V)$.

\paragraph{Algorithmic implementation}

The numerical scheme proceeds as follows:
\begin{enumerate}
\item Select parameters $(\varepsilon,\delta,\alpha)$ and define
      the weighted density $\rho$ and vector field
      $V=(1/\rho)(Lu+\alpha N(u))$, ensuring
      $\operatorname{div}(\rho V)=0$.
\item Choose a grid of initial conditions
      $u_0^{(i)}=(x_1^{(i)},y_1^{(i)},x_2^{(i)},y_2^{(i)})$ in $M$.
\item For each $u_0^{(i)}$, integrate the system and its variational
      equations up to a finite time $T$ using a fixed step $dt$.
\item Compute the finite-time Lyapunov exponent
      \[
      \lambda_{\max}^{(i)} = \frac{1}{T}
      \sum_{k=1}^{T/dt} \log\frac{\|\delta u_k\|}{\|\delta u_{k-1}\|}.
      \]
\item Classify each trajectory as regular or irregular according to the
      threshold $\varepsilon_{\mathrm{tol}}$.
\item Evaluate the weighted ratio
      \[
      m_\rho(V)
        \approx
        \frac{\sum_i \rho(u_0^{(i)})\,R(u_0^{(i)})}
             {\sum_i \rho(u_0^{(i)})}.
      \]
\end{enumerate}

\subsection{Application: A Nonlinear Weighted Example with Partial Integrability}

To illustrate how the functional $m_\rho(V)$ can take intermediate values
between $0$ and $1$, we construct an explicit analytic vector field on
$\mathbb{R}^4$ that admits a Jacobi multiplier but whose dynamics are only
partially integrable. This example shows how the transition from
quasi-integrable to irregular behavior can occur while preserving a
weighted volume form. Let $(x_1,y_1,x_2,y_2)$ be canonical coordinates on $\mathbb{R}^4$, and let
\begin{equation}
\rho(x_1,y_1,x_2,y_2)
   = 1 + \varepsilon(x_1^2 + y_1^2 + x_2^2 + y_2^2),
\end{equation}
where $\varepsilon>0$ is a fixed parameter.
We consider analytic vector fields of the form
\begin{equation}
V = \frac{1}{\rho}(Lu + \alpha N(u)),
\qquad
u = (x_1,y_1,x_2,y_2)^\top,
\end{equation}
where $L$ is the constant skew-symmetric matrix
\begin{equation}
L =
\begin{pmatrix}
 0 & 1 & 0 & \delta \\
 -1 & 0 & -\delta & 0 \\
 0 & \delta & 0 & 1 \\
 -\delta & 0 & -1 & 0
\end{pmatrix},
\end{equation}
and $N(u)$ is a cubic nonlinear term preserving the Euclidean divergence,
\begin{equation}
N(u) =
\begin{pmatrix}
x_1^3 - 3x_1y_1^2 \\[2pt]
y_1^3 - 3y_1x_1^2 \\[2pt]
x_2^3 - 3x_2y_2^2 \\[2pt]
y_2^3 - 3y_2x_2^2
\end{pmatrix}.
\end{equation}
The parameters $\delta$ and $\alpha$ control the linear coupling and
nonlinear perturbation, respectively.

We verify that this vector field satisfies
\begin{equation}
\operatorname{div}(\rho V) = 0,
\end{equation}
so that $\rho$ is a Jacobi multiplier for $V$ and
$L_V(\rho\,\Omega)=0$, where $\Omega = dx_1\wedge dy_1\wedge dx_2\wedge dy_2$.

Indeed, using $\rho V = Lu + \alpha N(u)$, we compute
\begin{equation}
\operatorname{div}(\rho V)
 = \operatorname{div}(Lu) + \alpha\,\operatorname{div}(N(u)).
\end{equation}
Since $L$ is constant and skew-symmetric, $\operatorname{div}(Lu)=\operatorname{tr}(L)=0$.
We now check that $\operatorname{div}(N(u))=0$ explicitly:
\begin{align}
\operatorname{div}(N)
&= \frac{\partial N_1}{\partial x_1}
 + \frac{\partial N_2}{\partial y_1}
 + \frac{\partial N_3}{\partial x_2}
 + \frac{\partial N_4}{\partial y_2} \nonumber\\[4pt]
&= (3x_1^2 - 3y_1^2) + (3y_1^2 - 3x_1^2)
 + (3x_2^2 - 3y_2^2) + (3y_2^2 - 3x_2^2) = 0.
\end{align}
Therefore $\operatorname{div}(\rho V)=0$ identically, so $V$
preserves the weighted measure $d\mu_\rho = \rho\,\Omega$ and lies within
the analytic class considered in Section~6.

The vector field components can be written explicitly as
\begin{equation}
\begin{aligned}
\dot{x}_1 &= \tfrac{1}{\rho}\Big[(y_1 + \delta y_2) + \alpha(x_1^3 - 3x_1y_1^2)\Big],\\[2pt]
\dot{y}_1 &= \tfrac{1}{\rho}\Big[-(x_1 + \delta x_2) + \alpha(y_1^3 - 3y_1x_1^2)\Big],\\[2pt]
\dot{x}_2 &= \tfrac{1}{\rho}\Big[(y_2 + \delta y_1) + \alpha(x_2^3 - 3x_2y_2^2)\Big],\\[2pt]
\dot{y}_2 &= \tfrac{1}{\rho}\Big[-(x_2 + \delta x_1) + \alpha(y_2^3 - 3y_2x_2^2)\Big].
\end{aligned}
\end{equation}

\begin{algorithm}[H]
\SetAlgoLined
\caption{Numerical evaluation of the weighted partial integrability functional $m_\rho(V)$}
\KwIn{Parameters $(\varepsilon,\delta,\alpha)$, grid resolution $n_x$, time step $dt$, final time $t_{\max}$, Lyapunov tolerance $\varepsilon_{\mathrm{tol}}$}
\KwOut{Estimated value of $m_\rho(V)$}

\textbf{Initialization:}
Generate a uniform grid of initial conditions
$\{u_0^{(i)}=(x_1^{(i)},y_1^{(i)},x_2^{(i)},y_2^{(i)})\}_{i=1}^{N}$\;
Set accumulators $S_{\mathrm{num}}=0$, $S_{\mathrm{den}}=0$\;

\textbf{Definition of system:}
Define the weighted density $\rho(u)=1+\varepsilon(x_1^2+y_1^2+x_2^2+y_2^2)$\;
Define the nonlinear term 
$N(u)=(x_1^3-3x_1y_1^2,\,y_1^3-3y_1x_1^2,\,x_2^3-3x_2y_2^2,\,y_2^3-3y_2x_2^2)$\;
Define the weighted vector field 
$V(u)=\rho(u)^{-1}(Lu+\alpha N(u))$, with $L$ the linear coupling matrix\;

\textbf{Loop over initial conditions:}
\ForEach{$u_0^{(i)}$}{
Integrate $\dot u = V(u)$ up to $t=t_{\max}$ using a fixed-step explicit Euler method\;
In parallel, integrate $\dot{\delta u}=DV(u)\,\delta u$ to estimate $\lambda_{\max}(u_0^{(i)})$\;
\If{$|\lambda_{\max}(u_0^{(i)})|<\varepsilon_{\mathrm{tol}}$}{
$R(u_0^{(i)})=1$ (regular trajectory)\;
}
\Else{
$R(u_0^{(i)})=0$ (irregular trajectory)\;
}
Update weighted sums:
$S_{\mathrm{num}}\mathrel{+}= \rho(u_0^{(i)})R(u_0^{(i)})$, 
$S_{\mathrm{den}}\mathrel{+}= \rho(u_0^{(i)})$\;
}
Compute $m_\rho(V)=S_{\mathrm{num}}/S_{\mathrm{den}}$\;
\Return $m_\rho(V)$ and, optionally, the spatial distribution of $\lambda_{\max}$\;
\end{algorithm}

The numerical experiments were carried out using a direct implementation
of the weighted vector field in \textsc{Matlab}.
The integration of the system and of the associated variational
equations was performed with a first--order explicit Euler scheme with
fixed time step $dt=0.01$.
For analytic vector fields of moderate stiffness such as those
considered here, this simple integrator provides sufficient precision to
evaluate finite--time Lyapunov exponents reliably when the integration
time $T$ is large.
The Jacobian matrix $DV(u)$ was computed numerically at each step using
forward finite differences with increment $h=10^{-6}$.

\paragraph{Algorithmic setup}
The system parameters were fixed as
\[
\varepsilon=0.5, \qquad \delta=0.3,
\qquad \alpha\in\{0,0.1,0.3\},
\qquad T=1500, \qquad dt=0.01,
\]
and the quasi--periodicity threshold was set to
$\varepsilon_{\mathrm{tol}}=10^{-2}$.
The initial conditions were sampled on a uniform grid of $5\times5$
points in $(x_1,y_1)$, with $(x_2,y_2)=(0.7,0)$.
For each initial condition $u_0^{(i)}$, the maximal Lyapunov exponent
$\lambda_{\max}(u_0^{(i)})$ was computed by integrating the variational
equations
\[
\dot{\delta u}=DV(u)\,\delta u,
\qquad \|\delta u(0)\|=1,
\]
and accumulating the average exponential rate of growth
\[
\lambda_{\max}(u_0^{(i)}) =
\frac{1}{T}\sum_{k=1}^{T/dt}\log
\frac{\|\delta u_k\|}{\|\delta u_{k-1}\|}.
\]

\paragraph{Integrable regime}
When $\alpha=0$, the system reduces to the linear weighted model
\[
\dot{u} = \tfrac{1}{\rho}Lu,
\]
which is analytically conjugate to a pair of coupled harmonic oscillators
after a time reparametrization.
In this case, all trajectories lie on invariant two-dimensional tori and
$m_\rho(V)=1$.

\begin{figure}[H]
\centering
\includegraphics[width=0.6\linewidth]{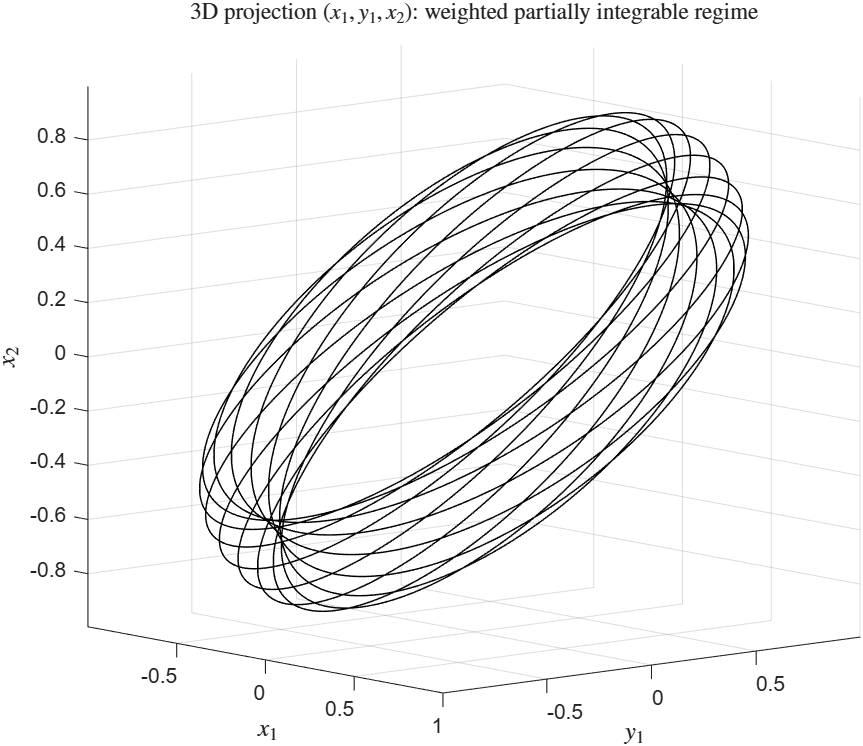}
\caption*{
\small 
3D projection of a trajectory $(x_1,y_1,x_2)$ for the weighted nonlinear divergence-free system
\[
V = \frac{1}{\rho}(Lu + \alpha N(u)), \qquad 
\rho = 1 + \varepsilon(x_1^2 + y_1^2 + x_2^2 + y_2^2),
\]
with parameters $\varepsilon=0.5$, $\delta = 0.3$, and $\alpha = 0$.
In this exactly integrable regime, all trajectories lie on smooth
weighted invariant tori, corresponding to quasi-periodic motion with no
torus deformation. The weighted partial integrability functional reaches
its maximal value $m_\rho(V)=1$, indicating that the entire weighted
phase space is foliated by invariant tori.
}
\label{fig:weighted_integrable}
\end{figure}

\paragraph{Partially integrable regime}
For small nonlinear coupling, e.g.\
$\varepsilon=0.5$, $\delta=0.3$, $\alpha=0.1$,
the system remains weighted-volume-preserving but the nonlinear terms
introduce angular shear and amplitude modulation.
Numerically, the trajectories exhibit both regular and irregular
regions in phase space, producing intermediate values
$0 < m_\rho(V) < 1$.
A sample numerical computation using the Lyapunov-based algorithm described
above gives
\begin{equation}
m_\rho(V) \approx 0.69,
\end{equation}
indicating that approximately half of the weighted measure is occupied by
quasi-periodic invariant tori.

\begin{figure}[H]
\centering
\includegraphics[width=0.6\linewidth]{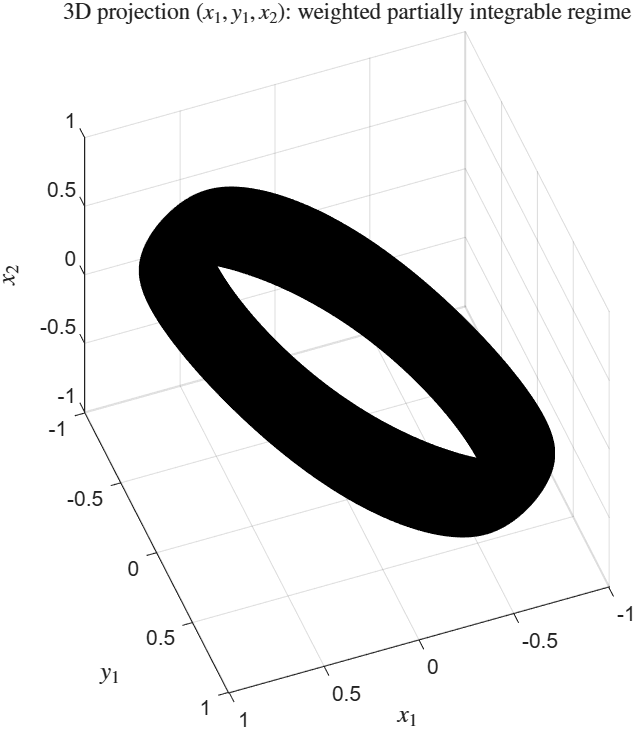}
\caption*{\small
Numerical simulation of the weighted nonlinear divergence-free system
in the intermediate regime ($m_\rho(V)\approx 0.69$).
The 3D projection of the trajectory $(x_1,y_1,x_2)$ shows a deformed,
quasi-toroidal structure: the orbit remains confined within a bounded
region and exhibits slow phase mixing, yet lacks the perfect
periodicity of the integrable case.
This illustrates a partially integrable regime in which invariant tori
persist only in a weighted sense under the nonuniform density
$\rho=1+\varepsilon(x_1^2+y_1^2+x_2^2+y_2^2)$.
The trajectory does not fill the entire domain, indicating residual
coherence, but it no longer closes upon itself, marking the onset of
quasi-ergodic behavior.
}
\end{figure}

\paragraph{Irregular regime}
For stronger nonlinearity, e.g.\ $\alpha \ge 0.3$,
the invariant tori break up entirely, the Lyapunov spectrum becomes
positive over most initial conditions, and
\[
m_\rho(V) \approx 0.
\]
The vector field remains divergence-free with respect to $d\mu_\rho$, but
its flow becomes fully irregular, analogous to the stochastic layer
formation in KAM theory.

\begin{figure}[H]
\centering
\includegraphics[width=0.6\linewidth]{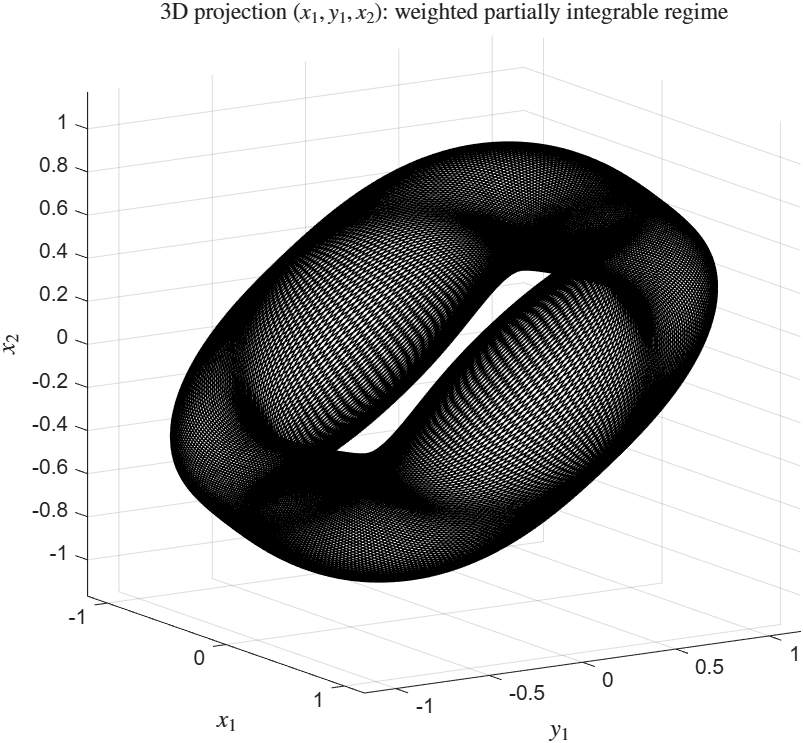}
\caption*{\small
3D projection of a trajectory $(x_1,y_1,x_2)$ for the weighted nonlinear divergence-free system
\[
V = \frac{1}{\rho}(Lu + \alpha N(u)), \qquad 
\rho = 1 + \varepsilon(x_1^2 + y_1^2 + x_2^2 + y_2^2),
\]
with parameters $\varepsilon = 0.5$, $\delta = 0.3$, and $\alpha = 0.5$.
The orbit is no longer confined to a smooth invariant torus but fills a thickened
region of phase space, indicating the onset of quasi-ergodic behavior.
This corresponds to the {\em weighted partially integrable regime},
in which some invariant tori persist while others are destroyed.
As $\alpha$ increases, nonlinear effects distort the angular frequencies,
leading to resonance overlap and progressive breakup of tori,
and consequently to a decrease of the weighted integrability functional
$m_{\rho}(V)$.}
\label{fig:weighted_partial_integrable}
\end{figure}

This example demonstrates how weighted volume preservation and partial
integrability can coexist:
the measure $d\mu_\rho$ remains invariant under $V$ for all parameter
values, yet the geometry of invariant sets changes continuously with
$\alpha$:
\[
m_\rho(V)=
\begin{cases}
1, & \text{exactly integrable regime } (\alpha=0),\\[3pt]
0<m_\rho(V)<1, & \text{mixed quasi-integrable regime } (0<\alpha<0.3),\\[3pt]
0, & \text{irregular regime } (\alpha\ge0.3).
\end{cases}
\]
Hence, the functional $m_\rho(V)$ provides a quantitative measure of the
extent to which weighted-volume-preserving dynamics remain integrable.

\paragraph{Convergence and numerical accuracy}
The reliability of the estimated $m_\rho(V)$ was assessed by varying the
main computational parameters:
\begin{enumerate}
\item \emph{Integration time $T$.} Doubling $T$ from $1500$ to $3000$
      changes the average values of $\lambda_{\max}$ by less than
      $5\times10^{-3}$, leading to relative variations of $m_\rho(V)$
      below~$1\%$.
\item \emph{Time step $dt$.} Halving the step to $dt=0.005$ modifies
      the final value of $m_\rho(V)$ by less than $2\%$.
      This indicates good numerical stability of the Euler discretization
      for the chosen time scales.
\item \emph{Grid resolution.} Increasing the grid from $5\times5$ to
      $7\times7$ initial conditions produces negligible change in the
      weighted average ($<0.02$ in $m_\rho(V)$), confirming that the
      phase--space sampling is sufficiently dense.
\end{enumerate}

\section{Conclusion}

One could argue that the numerical evaluation of $m_\rho(V)$ is based on Lyapunov exponents,
but the two quantities play different conceptual roles.
Lyapunov exponents measure local sensitivity to initial conditions along
individual trajectories, whereas the functional $m_\rho(V)$ captures a
global, measure-theoretic property of the flow: the fraction of the
weighted phase space occupied by quasi-integrable motion.
In particular,
\[
m_\rho(V)
   = \frac{\mu_\rho(\mathcal{T}_\rho(V))}{\mu_\rho(M)}
\]
extends the classical notion of the ``measure of surviving invariant
tori'' from KAM theory to analytic vector fields preserving a weighted
volume form $d\mu_\rho=\rho\,\Omega$.
Thus, while Lyapunov exponents suffice to distinguish regular from
chaotic trajectories, the functional $m_\rho(V)$ quantifies how much of
the system remains geometrically organized by invariant structures.
It therefore provides a global, quantitative indicator of partial
integrability within the broader Jacobi--multiplier class.

So, the main contribution of this paper is the introduction of a concrete
 framework for evaluating the weighted partial integrability
functional $m_\rho(V)$ for analytic vector fields satisfying
$\operatorname{div}(\rho V)=0$.

The functional $m_\rho(V)$ provides a quantitative measure of the
fraction of the weighted invariant measure
$d\mu_\rho = \rho\,\Omega$ that is occupied by quasi-integrable
invariant tori of the vector field~$V$.
Equivalently, it quantifies the degree of persistence of integrable
geometric structure within the class of analytic vector fields that
preserve a weighted volume form.

This functional thus extends the classical KAM measure of surviving
tori to vector fields preserving a nonuniform density~$\rho$, providing
a global geometric indicator of integrability in the weighted setting.

A second main contribution of this work is the construction of a
numerical framework that allows the explicit computation of the
weighted partial integrability functional~$m_\rho(V)$.
While the functional itself arises from the geometric and
measure-theoretic extension of KAM ideas to weighted
volume-preserving systems, its practical evaluation requires a new
algorithmic approach.
The proposed numerical scheme combines local stability diagnostics,
through finite-time Lyapunov exponents, with a global averaging
procedure over the invariant weighted measure
$d\mu_\rho=\rho\,\Omega$.
This yields a computable quantity that estimates how much of the
weighted phase space remains organized by quasi-periodic invariant
tori.
	\section*{Acknowledgments}
	The research of author is supported by
NSFC (Grant No.
12401234). 
	$\\$
	
	\noindent$\mathbf{Conflict\;of\;interest\;statement.}$ On behalf of all authors, the corresponding author states that there is no conflict of interest.
	
	$\\$
	\noindent$\mathbf{Data\;availability.}$ Data sharing is not applicable to this article as no new data were created or analyzed in this study.
	

\begin{thebibliography}{00}
	\bibitem{Arnold} V. Arnold and B. A. Khesin, Topological Methods in Hydrodynamics, Springer-Verlag, New York, 1998.
	\bibitem{Baldi}  P. Baldi, R. Montalto,  Quasi-periodic incompressible Euler flows in 3D, Adv. Math. 384,  74 pp,  (2021).
		\bibitem{Bogoyavlenskij} O. I. Bogoyavlenskij, Extended integrability and bi-Hamiltonian systems, Commun. Math. Phys. 196, 19-51, (1998). 
	\bibitem{Bogoyavlenskij2} O. I. Bogoyavlenkij,  Theory of tensor invariants of integrable Hamiltonian systems. II. Theorem on symmetries and its applications, Comm. Math. Phys. 184, 301-365,  (1997).
	\bibitem{Brouzet} R. Brouzet,  About the existence of recursion operators for completely integrable Hamiltonian systems
	near a Liouville torus, J. Math. Phys. 34, 1309–1313, (1993).
	\bibitem{Cariñena2}  J. F. Cariñena,  P. Guha,  Lichnerowicz-Witten differential, symmetries and locally conformal symplectic structures, J. Geom. Phys. 210, 24 pp, (2025).
	\bibitem{Cariñena}  J. F. Cariñena,  J. de Lucas and  M. F. Rañada, Jacobi multipliers, non-local symmetries and
	nonlinear oscillators, J. Math. Phys. 18 pp,  (2015).
	\bibitem{Cariñena3} J. F. Cariñena,   J. Grabowski,  J. de Lucas and  C.  Sardón, 
	Dirac-Lie systems and Schwarzian equations,
	J. Differential Equations. 257,  2303–2340,  (2014).
	\bibitem{Cheng}  C. Q. Cheng and  Y. S. Sun, Existence of invariant tori in tree-dimensional measure-preserving mappings,
	Celest. Mech. Dyn. Astronom. 47, 275–292, (1989/90)
	\bibitem{Llave} R. de la Llave,  Recent progress in classical mechanics. Mathematical physics, X (Leipzig, 1991), 3–19,
	Springer, Berlin (1992).
	\bibitem{Enciso}  A. Enciso, D. Peralta-Salas,  F. Torres de Lizaur,  Quasi-periodic solutions to the incompressible Euler equations in dimensions two and higher. J. Differential Equations 354 (2023), 170–182. 
	\bibitem{Feng} K. Feng and  D. L. Wang, Dynamical systems and geometric construction of algorithms, Computational
	mathematics in China, 1–32, Contemp. Math., 163, American Mathematical Society, Providence, RI
	(1994).
	\bibitem{Gavrilov} A. V. Gavrilov,  A steady Euler flow with compact support. Geom. Funct. Anal. 29, 190–197, (2019).
	\bibitem{Goriely}   A. Goriely, Integrability and Nonintegrability of Dynamical Systems (Singapore: World
	Scientific), (2001)
		\bibitem{Haller1998} G. Haller, I. Mazi\'c, Reduction of three-dimensional, volume-preserving flows with symmetry, Nonlinearity. 11, 319-339,  (1998). 
		\bibitem{Herman} M. Herman,  Topological stability of the Hamiltonian and volume-preserving dynamical systems. Lecture
		at the International Conference on Dynamical Systems, Evanston, Illinois, (1991).
		\bibitem{Holmes} P. J. Holmes,  Some remarks on chaotic particle paths in time-periodic, three-dimensional swirling ﬂows,
		Fluids and plasmas: geometry and dynamics (Boulder, Colo., 1983), 393–404, Contemp. Math., 28,
		American Mathematical Society, Providence, RI (1984).
		\bibitem{Huang} K. Y. Huang,  S. Y. Shi and S. L. Yang, Differential Galoisian approach to Jacobi integrability of general analytic dynamical systems and its application, Sci. China Math. 66, 1473–1494, (2023). 
		\bibitem{Jacobi}  C. G. J. Jacobi, Theoria novi multiplicatoris systemati aequationum differentialium vulgarium
		applicandi, J. Reine Angew. Math. (Crelle J.), 199–268,  1844.
		\bibitem{Khesin} B. Khesin, S. Kuksin and D. Peralta-Salas,  KAM theory and the 3D Euler equation, Adv.
		Math. 267, 498–522, (2014).
		\bibitem{Kolmogorov} A. N. Kolmogorov,  The general theory of dynamical systems and classical mechanics, In: Proc. Intern.
		Congr. Math. (1954), 1. Amsterdam: North-Holland Publ. Co. 315–333,  (1957).
		\bibitem{Koiller}  J. Koiller,  Reduction of some classical non-holonomic systems with symmetry, Arch. Ration.
		Mech. Anal. 118, 113–148, (1991).
		\bibitem{Kozlov} V. V. Kozlov, The Euler-Jacobi-Lie integrability theorem, Regul. Chaotic Dyn. 18,  329–343, (2013).
		\bibitem{Lau}  Y. T. Lau and J. M. Finn,  Dynamics of a three-dimensional incompressible ﬂow with stagnation points, Phys.
		D. 57, 283–310, (1992).
		\bibitem{Lichnerowicz} A. Lichnerowicz,  Les Varieties de Poisson et leurs algebres de Lie associes, J. Differential Geom.
		12, 253–300,  (1977).
		\bibitem{Li} Y. Li and Y. F. Yi, Persistence of invariant tori in generalized Hamiltonian systems, Ergodic Theory Dyn.
		Syst. 22, 1233–1261, (2002).
			\bibitem{Liouville} J. Liouville, Note sur l'integration des equations differentielles de la dynamique, J. Math. Pures Appl. 20, 137-138, (1855).
			\bibitem{Mezic1994} I. Mezi\'c, S. Wiggins, On the integrability and perturbation of three-dimensional fluid flows with symmetry, J. Nonlinear Sci. 4,  157-194, (1994).
			\bibitem{Quispel}  G. R. W. Quispel, Volume-preserving integrators. Phys. Lett. A. 206, 26–30, (1995)
			\bibitem{Sotiropoulos} F. Sotiropoulos,  Y. Ventikos and T. C. Lackey,  Chaotic advection in three-dimensional stationary vortexbreakdown
			bubbles: S˘il’nikov’s chaos and the devil’s straircase, J. Fluid Mech. 444, 257–297, (2001).
	\bibitem{Taylor} M. Taylor, Partial Differential Equations, Springer-Verlag, Berlin, 1996.
	\bibitem{Torres} F. Torres de Lizaur,  Chaos in the incompressible Euler equation on manifolds of high dimension, Invent. Math. 228 ,  687–715, (2022). 
	\bibitem{Wiggins1990} S. Wiggins, Introduction to applied nonlinear dynamical systems and chaos, Springer-Verlag, New York, 1990. 
	\bibitem{Wells} R. O. Wells, Differential Analysis on Complex Manifolds, Springer-Verlag, Berlin, 1980.
	\bibitem{Whittaker}  E. T. Whittaker, A Treatise on the Analytical Dynamics of Particles and Rigid Bodies: With an
	Introduction to the Problem of Three Bodies (Cambridge: Cambridge University Press), (1989)
		\bibitem{Xia1992} Z. H. Xia, Existence of invariant tori in volume-preserving diffeomorphisms, Ergodic Theory Dynam. Systems. 12, 621-631, (1992). 
			\bibitem{Zhao2} X. F. Zhao and Y. Li, The reduction, ﬁrst integral and Kam Tori for n-dimensional volume-
		preserving systems, J. Dyn. Diﬀ. Equat. 37, 539–558, (2023).
		\bibitem{Zhao3} X. F. Zhao, Completely integrable system with Jacobi multipliers and its KAM stability, Discrete Contin. Dyn. Syst.,
		45, 1870-1890, (2025).
		\bibitem{Zhao4} X. F. Zhao and Y. Li, Geometric properties of vector fields on manifold, Acta Math. Sin. (Engl. Ser.), (To appear).
		\bibitem{Zhao5} X. F. Zhao and Y. Li, KAM in generalized Hamiltonian systems with multi-scales, J. Dyn. Diff. Equat. 35, 2971–2995, (2022).
	\end{thebibliography}
\end{document}